\documentclass[runningheads]{llncs}
\usepackage{graphicx}
\usepackage[subtle]{savetrees}
\usepackage{amssymb}
\usepackage{lmodern}        
\usepackage[T1]{fontenc}  
\usepackage[utf8]{inputenc} 

\usepackage{subfig} 

\usepackage{placeins} 
\usepackage{booktabs} 
\usepackage{multirow} 
\usepackage{xcolor}

\usepackage{paralist} 
\usepackage{verbatim} 

\usepackage{array} 
\usepackage{mathtools} 
\usepackage{amssymb} 
\usepackage{algorithm}
\usepackage{thmtools,thm-restate}
\usepackage[noend]{algpseudocode}

\newcommand{\tol}[1]{\lceil #1\rceil}
\newcommand{\finv}{f^{inv}}
\newcommand{\fem}{f^{\emptyset}}

\newif\iflong

\begin{document}
\title{On tolerance of discrete systems with respect to transition perturbations\thanks{This work has been supported by the National Science Foundation under NSF SaTC awards CNS-1801342 and CNS-1801546.
	The proofs of our results can be found in \cite{Meira-Goes2021-tolerance}.
	}}
\titlerunning{On tolerance of discrete systems}
%
\author{R\^omulo Meira-G\'oes\inst{1,2} \and
Eunsuk Kang\inst{1} \and
St\'ephane Lafortune\inst{2} \and\\
Stavros Tripakis\inst{3}
}
\authorrunning{R. Meira-G\'oes et al.}
\institute{School of Computer Science,  Carnegie Mellon University, Pittsburgh USA \email{\{rmeirago,eunsukk\}@andrew.cmu.edu}\and
Dept. of Elect. Eng. and Computer Science,  University of Michigan, Ann Arbor USA
\email{\{romulo,stephane\}@umich.edu}\\
\and
Khoury College of Computer Science, Northeastern University, Boston USA\\
\email{stavros@northeastern.edu}}
\maketitle              
\begin{abstract}
Control systems should enforce a desired property for both expected/modeled situations as well as unexpected/unmodeled environmental situations.
Existing methods focus on designing controllers to enforce the desired property only when the environment behaves as expected.
However, these methods lack discussion on how the system behaves when the environment is perturbed.
In this paper, we propose an approach for analyzing control systems with respect to their tolerance against environmental \emph{perturbations}.
A control system tolerates certain environmental perturbations when it remains capable of guaranteeing the desired property despite the perturbations.
Each controller inherently has a level of tolerance against environmental perturbations.
We formally define this notion of tolerance and describe a general technique to compute it, for any given regular property.
We also present a more efficient method to compute tolerance with respect to invariance properties.
Moreover, we introduce and solve new controller synthesis problems based on our notion of tolerance. 
We demonstrate the application of our framework on an autonomous surveillance example.
\end{abstract}

\graphicspath{{Figs/}}
\section{Introduction}
In control systems, a controller is designed to enforce a desired property over the environment it controls.
For example, the cruise control system enforces the car (environment) to maintain a desired speed (desired property).
In this context, classical reactive synthesis methods provide means to synthesize controllers that correctly assure a desired property expressed in formal logic \cite{Ramadge:1987,Pnueli:1989a,Lafortune:2008,Tabuada:2009,Belta:2017}.
However, these methods heavily rely on assumptions about the behavior of the environment.
For instance, the environment in the cruise control system is specified by the dynamics of the car, the road, the meteorological conditions, etc.
Thus, the correct behavior of the controlled system is only guaranteed under these environmental assumptions/models.

Perturbations from the assumed environmental model jeopardize the correctness of the controller, limiting the application of reactive synthesis methods.
Therefore, in addition to correctness, controllers should be designed to tolerate \emph{reasonable} model perturbations. 
Even when the environment behaves unexpectedly, the controlled system must be correct.

Perturbations can be introduced by the designer prior to the synthesis procedure becoming part of the environmental model \cite{Bloem:2009,Majumdar:2011,Topcu:2012,Zhou:1998}.
Therefore, it is possible to use the existing approaches to synthesize a controller that is tolerant against certain perturbations.
 However, given an existing controller, it's not clear how tolerant this controller is and whether it's actually tolerant enough against some set of perturbations under consideration.

A different type of analysis is to pose the following question:
\emph{Which environmental perturbations can an existing controller tolerate?}
Each controller inherently has a \emph{level of tolerance} against environmental perturbations whether or not the designer is aware of it.
Explicitly computing this level of tolerance is useful for the designer in many ways, e.g., it helps the designer to decide if it is safe to deploy the controller or whether a new controller is needed.

In this paper, we investigate the tolerance of controllers for \emph{discrete transition systems}.
We model perturbations as additional transitions to the original model, creating a framework to analyze their impact in the controlled system.
The controller is tolerant against a perturbation if the perturbed controlled system still satisfies a given desired property.
We go on to define a new notion of controller tolerance against perturbations as the set of all perturbations for which the controller is tolerant.
Based on this new notion, we define the problem of computing the controller tolerance given a desired regular property.
We show that this problem can be reduced to a sequence of model checking problems for discrete systems.

To more efficiently solve the computation problem above, we investigate the notion of controller tolerance with respect to invariance properties.
In this case, we show that a single perturbation represents all perturbations for which the controller is tolerant.
This result allows us to reduce the computation problem to a reachability analysis problem.
It also allows us to investigate and solve three controller synthesis problems: synthesis of the most tolerant controller, of the least tolerant controller, and of a controller that achieves a minimum tolerance threshold.

We have four main contributions in this paper:
\vspace*{-.1cm}
\begin{itemize}
	\item We define a new notion of controller tolerance and a general technique to automatically compute it;
	\item We investigate this notion of tolerance with respect to invariance properties and devise an efficient algorithm to compute it;
	\item We propose three new controller synthesis problems and provide their solution based on existing reachability analysis techniques;
	\item We provide a prototype implementation of our algorithms and demonstrate our approach on an example involving surveillance protocols.
\end{itemize}

\section{Motivating example}\label{sect:motivating}
As a motivating example, we consider a surveillance scenario of two autonomous drones, $ego$ and $srv$.
These drones monitor the surroundings of a building as depicted in Fig.~\ref{fig:mot_dist}.
$Ego$ desires to obtain information about the building without being captured by $srv$, i.e., where ``captured" means both that drones are in the same location.
It also assumes that $srv$ surveils the building by following the strategy depicted in Fig.~\ref{fig:srv-str}, i.e., $srv$ surveils the building by always moving in the clockwise direction.

\begin{figure}[thpb]
\centering  
\subfloat[Surveillance overview]  
{  
\centering
\includegraphics[width=0.25\columnwidth]{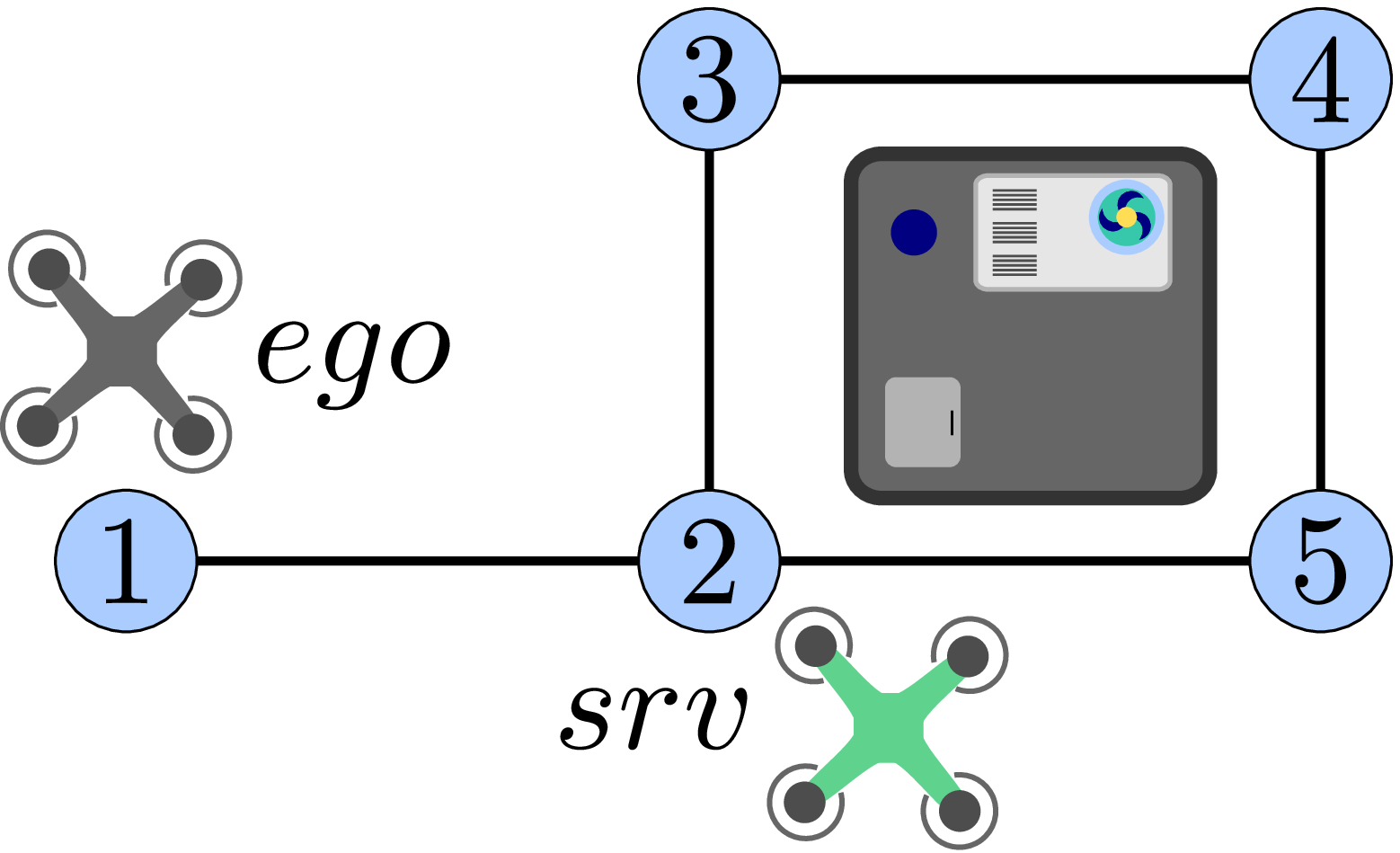}
\label{fig:mot_dist}
}
\qquad
\subfloat[$srv$ strategy assumed by $ego$]  
{  
\centering
\includegraphics[width=0.25\columnwidth]{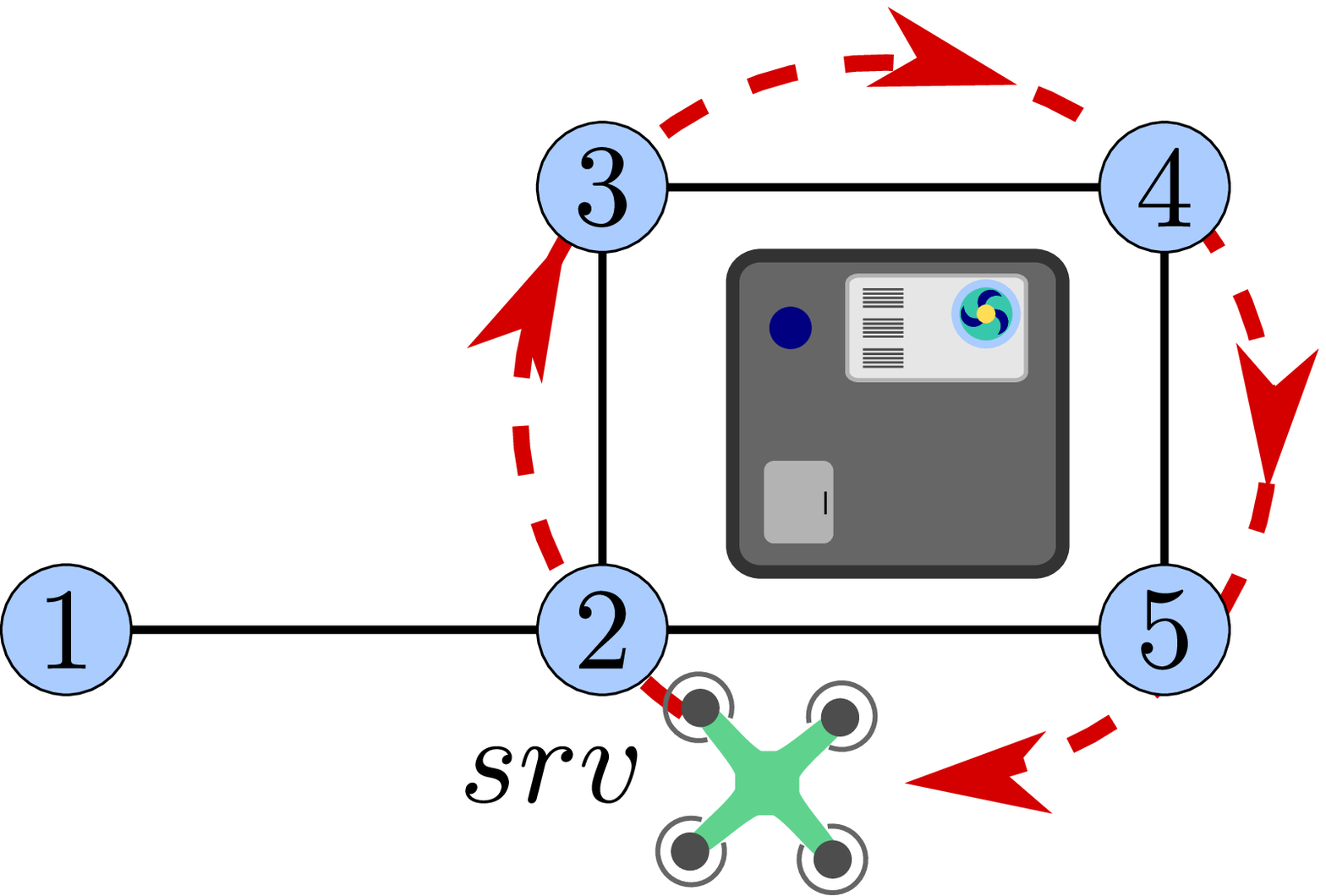}
\label{fig:srv-str}
}
\caption{Motivating example of a surveillance scenario}
\end{figure}

\vspace*{-.2cm}
Classical reactive synthesis techniques can synthesize a controller for $ego$ that guarantees the satisfaction of its property \cite{Pnueli:1977,Pnueli:1989a,Alur:2001,Bloem:2012,Gradel:2002}.
For example, we can synthesize two controllers that guarantee that $srv$ does not capture $ego$: controller 1 maintains $ego$ most of the time in location $1$ but it allows $ego$ to visit location 2 when $srv$ is in location 4, and controller 2 ensures that $ego$ trails two steps behind $srv$.
These controllers, however, guarantee the property \emph{only for the model of the system described above.}

Now, suppose that $srv$ does not conform with the strategy assumed by $ego$, e.g., $srv$ decides to go counter-clockwise to monitor the building.
Because of this perturbation, the two synthesized controllers might not guarantee the property.
Hence, to validate these controllers against these perturbations, an extra verification step must be performed.
For example, if $srv$ moves counter-clockwise, controller 1 guarantees that $ego$ is not captured but controller 2 does not.
However, one needs to verify that indeed these controllers continue or not to satisfy this property.
Another option is to synthesize a new controller based on the model of the system augmented with the possible ``known'' perturbations \cite{Topcu:2012}.

In comparison, as was mentioned in the introduction, we pose the following question:
\emph{For which model perturbations can a controller ensure the given property?}
Our notion of tolerance explicitly states that controller 1 ensures $ego$'s property even when $srv$ moves counter-clockwise, whereas controller 2 does not.
It also allows us to affirm that controller 1 is more tolerant than controller 2.
In general, our notion of tolerance enables analyzing the environmental perturbations a controller tolerates.
Moreover, it also enables comparing different controllers to determine which of them may be more tolerant as well as synthesizing controllers with a given level of tolerance.

\section{Preliminaries}\label{sect:preliminaries}
This section describes the underlying formalism used to model the environment, feedback controlled systems, and the properties enforced by them.

\subsubsection*{Labeled transition systems}
In this work, we use labeled transition systems to model the behavior of the environment.
\begin{definition} 
A \emph{labeled transition system} (LTS) $T$ is a tuple $\langle Q, Act, R, I\rangle$, where $Q$ is a finite set of states, $Act$ is a finite set of actions, $R\subseteq Q\times Act\times Q$ is the transition relation of $T$, and $I\subseteq Q$ is a nonempty set of initial states. 
\end{definition}
Let $Post_T(q,a)$ denote the set of states reachable from state $q\in Q$ and action $a\in Act$, i.e., $Post_T(q,a) := \{q'\in Q\mid (q,a,q')\in R\}$.
A \emph{run} of $T$ starts in an initial state in $I$ and is followed by a finite or infinite alternating sequence of actions and states complying with transitions in $R$, e.g., $x_0a_0x_1a_1\dots x_n$ such that $x_{i+1}\in Post_T(x_i,a_i)$ for all $i<n$ and $x_0 \in I$.
The set of all runs in $T$ is denoted by $Runs(T).$
A \emph{path} of $T$ is the sequence of states in a run of $T$, e.g., for $x_0a_0x_1\in Runs(T)$, then $x_0x_1$ is a path of $T$.
The sets $Paths_{fin}(T)$ and $Paths(T)$ denote the set of finite and all paths in $T$, respectively.

\begin{example}\label{example:motivating}
We model the motivating example in Section~\ref{sect:motivating} using LTS.
The states represent the discrete locations of $ego$, $\{1,2,3,4,5\}$ and $srv$, $\{2,3,4,5\}$.
The possible actions of the system consist of $ego$ selecting its desired next location, i.e., $Act = \{m1,\dots, m5\}$ where $mi$ means that $ego$ moves to location $i$.  
The transition relation is defined by a few update rules and assumptions.
The two drones move synchronously to their next location.
Next, both drones can only move to locations that are connected by an edge in Fig.~\ref{fig:mot_dist}.
Lastly, we assume that $srv$ surveils the building using the strategy defined in Fig.~\ref{fig:srv-str}, e.g., $srv$ moves to location $2$ when $5$ is its current location.
The system is initialized in state $(1,5)$, i.e., $ego$ in location $1$ and $srv$ in location $5$.
Figure~\ref{fig:dist-LTS} partially depicts the LTS $T$ defined by this example.
\end{example}
\begin{figure}[!h]
\centering  
\subfloat[Partial LTS $T$ of the surveillance example. States are of the form $(ego$ location$,\ srv$ location$)$ and edge labels represent the actions of $ego$. Transitions in blue are the missing transitions in this partial LTS. State $(2,2)$ is in red since $srv$ captured by $ego$.]  
{  
\centering
\includegraphics[width=0.5\columnwidth]{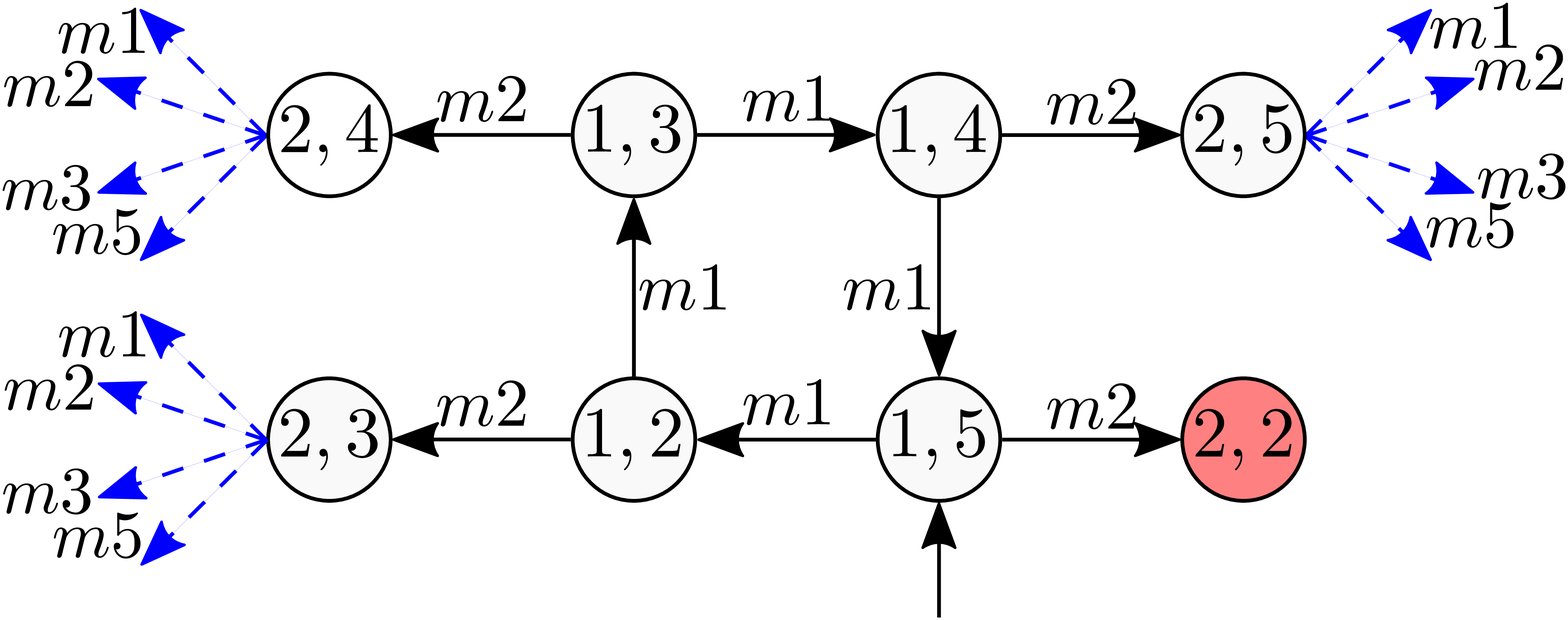}
\label{fig:dist-LTS}
}
\qquad
\subfloat[LTS representation of $T|f$]  
{  
\centering
\includegraphics[width=0.16\columnwidth]{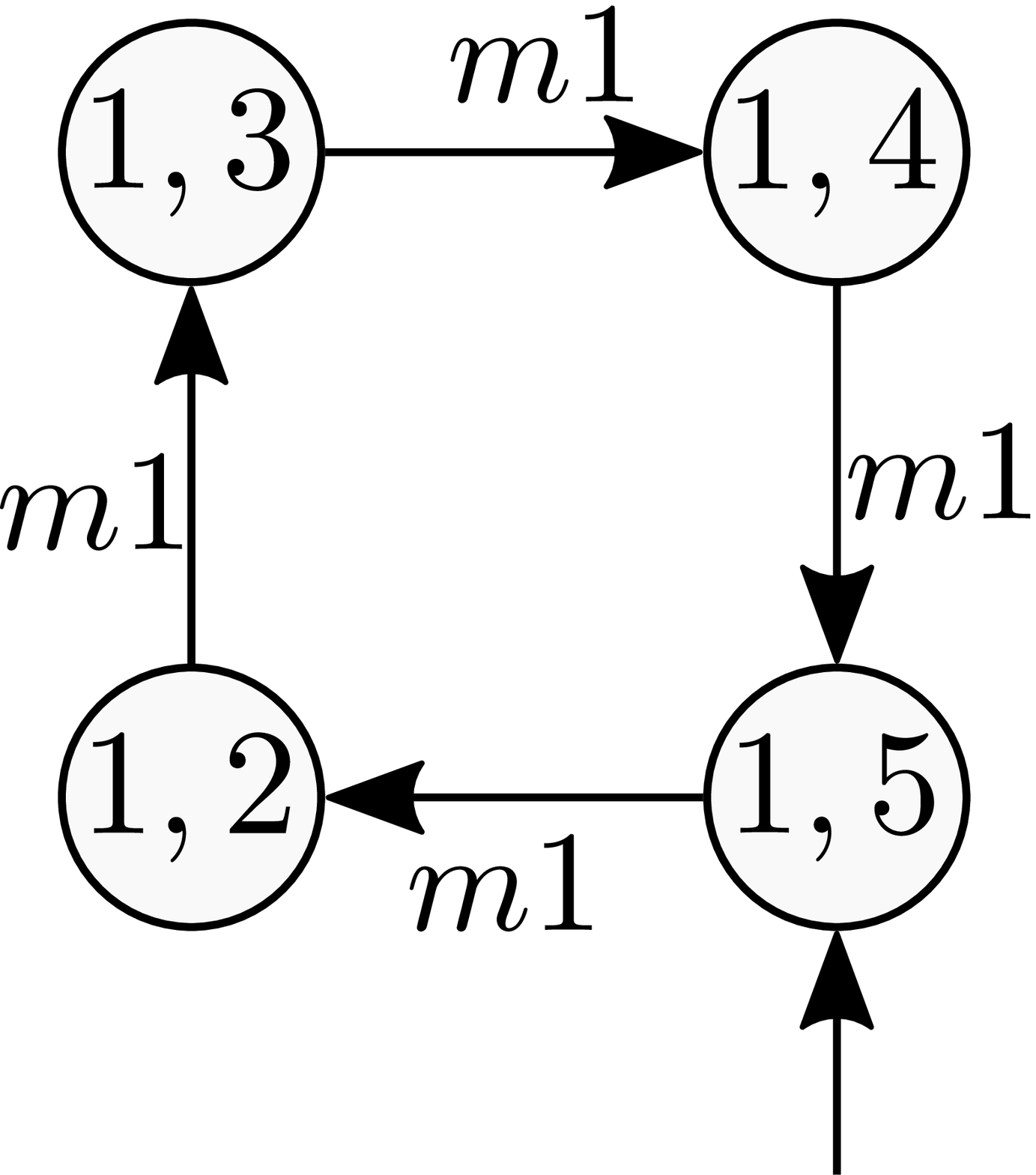}
\label{fig:mot-ctr}
}
\caption{LTS motivating example and controlled system}
\end{figure}

\begin{remark}
Our definition of LTS assumes that the set of actions $Act$ are controllable actions, e.g., $ego$ selects action $m1$.
However, the nondeterministic transition relation encapsulates uncontrollable actions.
Intuitively, after an action is selected, the environment decides which state the system will be in.
\end{remark}

Given a finite set $A$, the usual notations $|A|$, $A^*$, $A^+$, and $A^{\omega}$ denote the cardinality of $A$, the set of all finite sequences, the set of all non-empty finite sequences, and the set of all infinite sequences of elements in $A$, respectively.
For convenience, we write $x_{0\dots n}$ for any finite sequence of states $x_0\dots x_n$.

\subsubsection*{Control strategy}
Given an LTS $T$, a control strategy, or simply \emph{controller}, for $T$ is a function that maps a finite sequence of states to a set of actions, i.e., $f:Q^{+}\rightarrow 2^{Act}$.
A \emph{controlled run} in $T$ is a run of $T$ where actions are constrained by controller $f$, e.g., $x_0a_0\dots\in Runs_{inf}(T)$ such that $a_{i} \in f(x_{0\dots i})$ for any $i\geq 0$.
The set of all controlled runs, denoted by $Runs(T|f)$, defines the closed-loop system of $f$ controlling $T$.
For convenience, this closed-loop system is denoted by $T|f$.
The sets of finite and all controlled paths are denoted by $Paths_{fin}(T|f)$ and $Paths(T|f)$.
A controller has \emph{finite memory} if its decisions depend only on a finite number of states.
It is \emph{memoryless} if its decisions depend on the last state, $f:Q\rightarrow 2^{Act}$.
When $f$ has finite memory, $T|f$ can be represented by an LTS.



\begin{example}\label{example:controller}
Back to our motivating example, we give an example of a simple memoryless controller that is set to maintain $ego$ in location $1$.
Formally, the controller is defined as $f(1,i) = f(2,i) = \{m1\}$ for $i\in \{2,\dots, 5\}$, $f(3,4) = f(5,4) = \{m2\}$, $f(3,5) = f(4,5) = \{m3\}$, $f(3,2) = f(4,2) = f(5,2) = \{m4\}$, $f(4,3) = f(5,3) = \{m5\}$, otherwise $f(q) = \emptyset$.
Figure~\ref{fig:mot-ctr} shows the reachable states of the LTS representation of $T|f$ when $I = \{(1,5)\}$.  
\end{example}

\subsubsection*{Property}

In this work, we consider the class of linear-time (LT) properties over the set of states $Q$ of a given LTS $T$ \cite{Baier:2008}, i.e., property $P$ is a subset $P \subseteq Q^{\omega}\cup Q^{*}$.
In words, an LT property is a set of infinite and finite sequences of states that represents an ``admissible/desired'' set of paths of $T$. 
An LTS $T$ satisfies property $P$, notation $T\models P$, whenever $Paths(T)\subseteq P$.
Similarly, a controlled system $T|f$ satisfies property $P$ if $Paths(T|f)\subseteq P$.

\section{Tolerance against perturbations} \label{sect:tolerant-notion}
\subsection{Perturbations}
Model-based control theory methods are grounded on a model of the environment under control.
This model is always an approximation of the true system.
For this reason, we must take into account possible mismatches between the model of the environment and the true environment when designing a controller.
In the case of LTS, we model these possible mismatches, called \emph{perturbations}, as additional transitions.
Formally, \emph{a perturbation}\footnote{For simplicity, we define perturbations without removing the transition relation $R$ as to not overload our definitions with the removal of $R$. All of our results hold when perturbations are defined as $d\subseteq (Q\times Act\times Q)\setminus R$.} is a set of transitions $d\subseteq (Q\times Act\times Q)$.

For example, transition $(1,2),m1,(1,5)$ represents a perturbation in $srv$ assumed clockwise strategy depicted in Fig.~\ref{fig:srv-str}.
$Srv$ decides to go back to position $5$ instead of going to position $3$.
A second type of perturbation is transition $(1,2),m1,(2,3)$ where $ego$ gets pushed to location $2$ even though it has selected an action to stay in location $1$.

Given a perturbation set, we can define the \emph{perturbed system} by augmenting the transition relation of the LTS with the perturbation set. 

\begin{definition}
Let an LTS $T = \langle Q, Act, R, I\rangle$ and a perturbation $d \subseteq Q\times Act \times Q$ be given.
We define the \emph{perturbed system} $T_{d}$ as $T_{d}:=\langle Q, Act, R\cup d, I\rangle$.
\end{definition}

A controller $f$ that guarantees property $P$ for system $T$, $T|f\models P$, might violate this property for the perturbed system $T_d$. 
Thus, one needs to check if $f$ continues to satisfy $P$ for $T_d$, i.e., if $T_{d}|f\models P$.

\begin{definition}\label{def:tolerable-dev}
Controller $f$ is a \emph{tolerant controller} with respect to LTS $T$, perturbation $d$, and property $P$ if $T_{d}|f\models P$.
Perturbation $d$ is a \emph{tolerable perturbation} with respect to $T$, $f$, and $P$ if $f$ is a tolerant controller with respect to $T$, $d$, and $P$.
\end{definition}

\subsection{Comparing perturbations}
Given perturbations $d_1$ and $d_2$ such that $d_1\subseteq d_2$, $d_2$ perturbs LTS $T$ more than $d_1$ since $Runs(T_{d_1})\subseteq Runs(T_{d_2})$.
Our definition of tolerable perturbations takes into account not only the perturbed system, but a controller $f$ and its controlled behavior, e.g., $T_{d_1}|f$.
By including the controller to close the loop, two incomparable perturbations can generate comparable set of runs, i.e., it might be that $Runs(T_{d_1}|f)\subseteq Runs(T_{d_2}|f)$ even when $d_1\not\subseteq d_2$ and $d_2\not\subseteq d_1$.
In this scenario, $d_2$ perturbs the controlled system more than $d_1$ since $d_2$ has more influence on the controlled behavior.
Moreover, whenever $d_1\subseteq d_2$, it follows that $Runs(T_{d_1}|f)\subseteq Runs(T_{d_2}|f)$ for any controller $f$.
Based on this discussion, we define what it means to a perturbation be more or less ``powerful'' than other.

\begin{definition}\label{def:dev-order}
Let an LTS $T$, controller $f$, and perturbations $d_1$ and $d_2$ be given.
We say $d_1$ is \emph{at least as powerful} as $d_2$ with respect to $f$, denoted by $d_2\preceq_f d_1$, if
\begin{enumerate}
\item[(i)] $Runs(T_{d_2}|f)\subset Runs(T_{d_1}|f)$, or;
\item[(ii)]  $Runs(T_{d_2}|f) = Runs(T_{d_1}|f)\Rightarrow d_2\subseteq d_1$.
\end{enumerate}
Whenever the controller $f$ is clear from the context, we write $\preceq$ instead of $\preceq_f$.
\end{definition}

Intuitively, a perturbation $d_1$ is at least as powerful as perturbation $d_2$ with respect to controller $f$, if the controlled perturbed system $T_{d_1}|f$ can generate any run that $T_{d_2}|f$ can, and if the two controlled systems generate exactly the same set of runs, then $d_2\subseteq d_1$.
It follows that the ordering $\preceq$ forms a partial order over the set of perturbations of $T$.
To provide more intuition on $\preceq$, we have the following example.

\begin{example}\label{example:perturbation-relation}
Consider the LTS $T$ shown in Fig.~\ref{fig:T-cx-inv-ctr} and the property defined by all sequence of states that do not reach state $3$, e.g., the sequence $143$ violates this property.
We define the memoryless controller $f$ as $f(q) = \{b\}$ if $q\neq 3$, and $f(3) = \emptyset$.
It follows that $f$ satisfies the stated property, i.e., $T|f\models P$.

Consider the tolerable perturbations $d_1 = \{(1,b,2)\}$, $d_2 = \{(1,b,4)\}$, $d_3 = \{(2,b,3)\}$, and $d_4 = \{(4,b,3)\}$.
Perturbations $d_1$ and $d_2$ are at least as powerful as $d_3$ and $d_4$, i.e., $d_3\preceq d_1$, $d_4\preceq d_1$, $d_3\preceq d_2$, and $d_4\preceq d_2$.
On the other hand, $d_1$ and $d_2$ are incomparable with respect to $\preceq$ as their perturbed controlled systems generate incomparable runs.
Perturbations $d_3$ and $d_4$ are also incomparable even though $Runs(T_{d_3}|f) = Runs(T_{d_4}|f)$.
In this case, condition (ii) in Def.~\ref{def:dev-order} is violated as $d_3\not\subseteq d_4$ and $d_4\not\subseteq d_3$.
\begin{figure}[!h]
\centering  
\subfloat[LTS $T$]  
{  
\centering
\includegraphics[width=0.16\columnwidth]{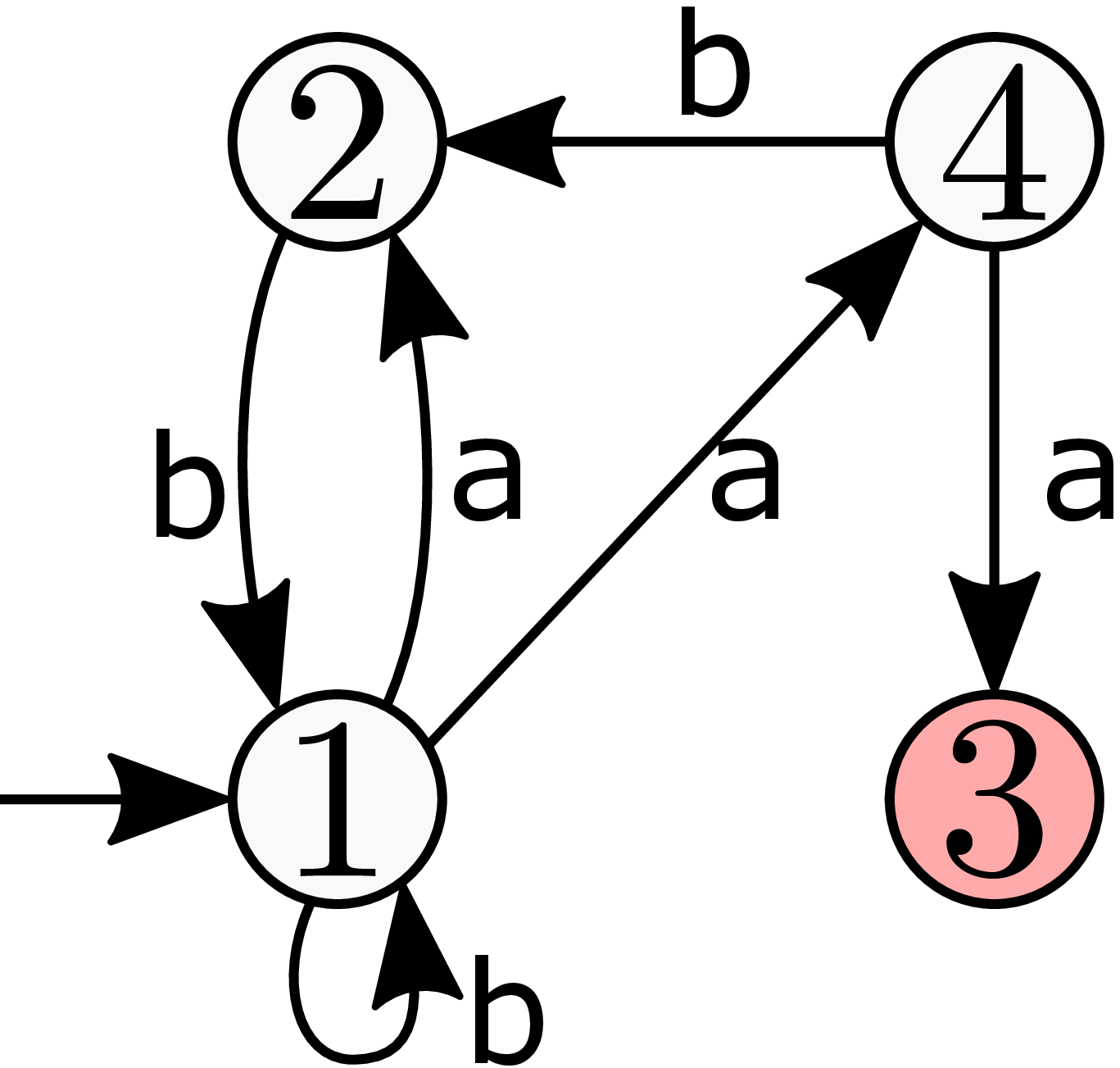}
\label{fig:T-cx-inv-ctr}
}
\qquad
\subfloat[LTS $T_{d_1}$]  
{  
\centering
\includegraphics[width=0.16\columnwidth]{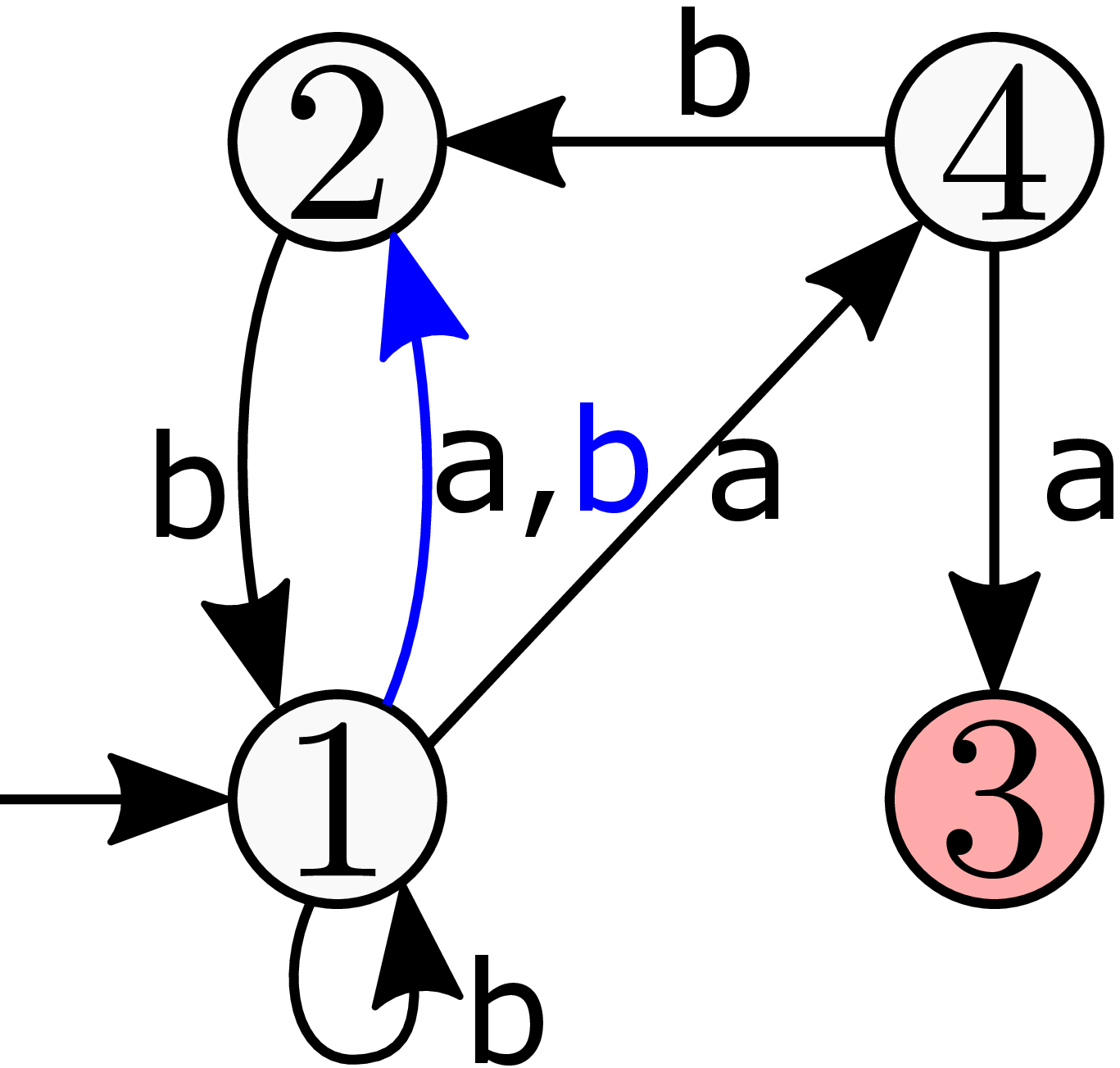}
\label{fig:T-cx-d1}
}
\qquad
\subfloat[LTS $T_{d_3}$]  
{  
\centering
\includegraphics[width=0.16\columnwidth]{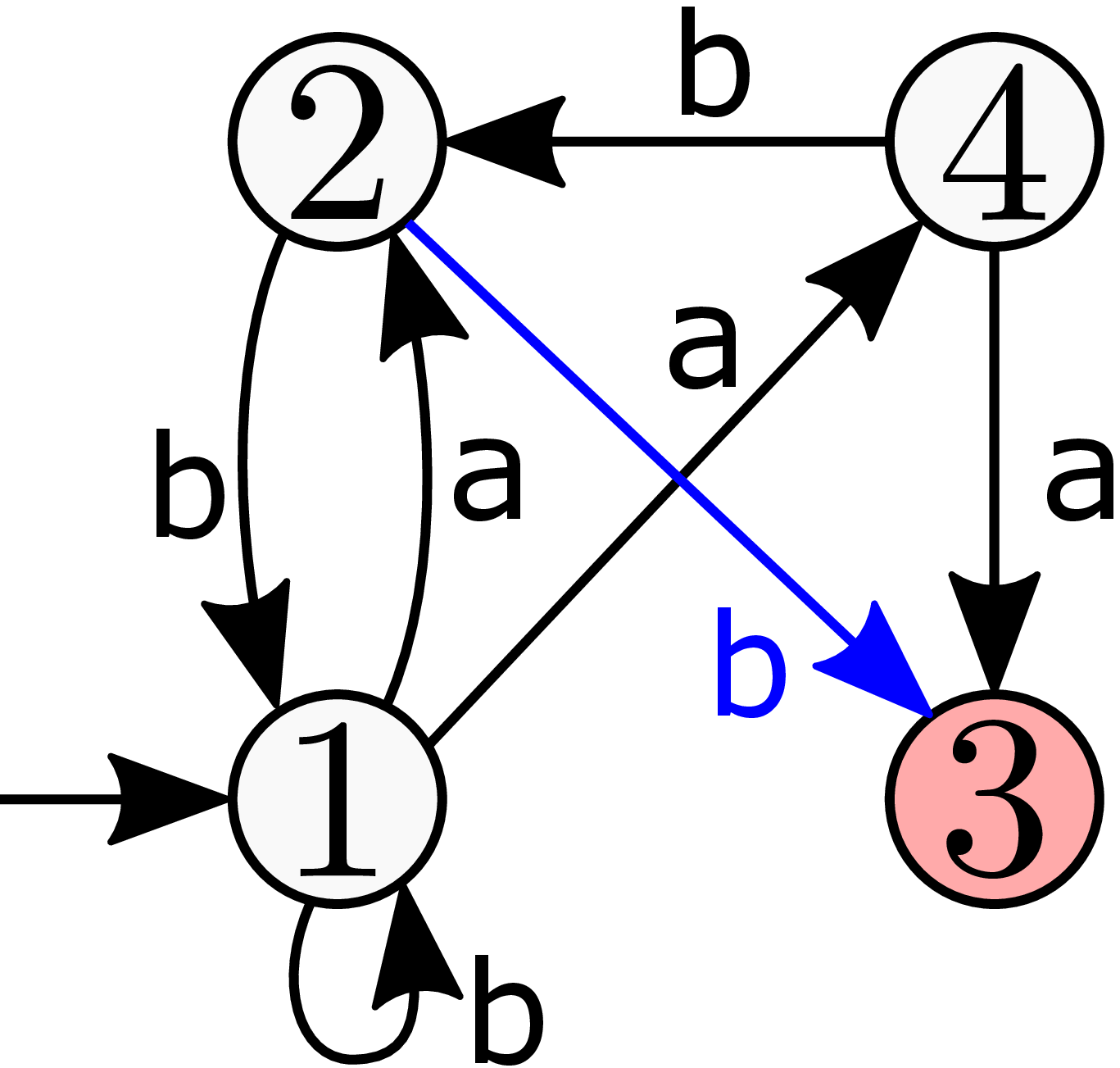}
\label{fig:T-cx-d3}
}
\caption{Tolerable perturbations}
\end{figure}
\end{example}

\vspace{-.6cm}
In Example~\ref{example:perturbation-relation}, perturbations $d_3$ and $d_4$ are incomparable with respect to $\preceq$ even though $Runs(T_{d_3}|f) = Runs(T_{d_4}|f)$.
Fortunately, their union, $d_3\cup d_4$, generates the same controlled runs and it is at least as powerful as $d_3$ and $d_4$.
This result establishes the existence of a maximal perturbation within the set of perturbations that generate the same controlled runs.

\begin{restatable}{proposition}{propmaxdev}
Given LTS $T$, controller $f$, and perturbations $d_1,d_2$ such that $Runs(T_{d_1}|f) = Runs(T_{d_2}|f)$, it follows that $d_1,d_2\preceq_f d_1\cup d_2$ and $Runs(T_{d_1\cup d_2}|f) = Runs(T_{d_1}|f) = Runs(T_{d_2}|f)$.
\end{restatable}

\subsection{Tolerance definition}

Prior work on robustness of discrete transition systems assumes that perturbation $d$ is given and one checks the tolerance of controller $f$ with respect to $d$ and property $P$.
Our approach transforms the assumption of given perturbations into our object of study.
Intuitively, we search for all possible tolerable perturbations $d$ with respect to LTS $T$, controller $f$, and property $P$.

\begin{definition} \label{def:tolerance}
Let LTS $T$, property $P$, and controller $f$ such that $T|f\models P$ be given.
The tolerance of $f$ with respect to $P$ and $T$, denoted as $\Delta(T,f,P)$, is a collection of perturbations ($\Delta(T,f,P)\subseteq 2^{Q\times Act\times Q}$) such that:
\begin{enumerate}
\item $\forall d \in \Delta(T,f,P).\ T_{d}|f\models P$  \emph{[$d$ is tolerable]};
\item $\forall d\subseteq Q\times Act\times Q .\ T_{d}|f\models P\Rightarrow\exists d'\in \Delta(T,f,P).\ d\preceq d'$ \emph{[$d$ is represented]}; 
	\item $\forall d,d'\in \Delta(T,f,P).\ d\neq d'\Rightarrow d \not\preceq d'$ \emph{[unique representations]}.
\end{enumerate}
\end{definition}

Conditions 2 and 3 in Def.~\ref{def:tolerance} enforce that only maximal tolerable perturbations with respect to $\preceq$ are in $\Delta$.
Formally, the set $\Delta$ defines an antichain, with respected to $\preceq$, of maximal tolerable perturbations.
Intuitively, the set $\Delta$ defines an upper bound on the possible perturbations from $T$ that controller $f$ tolerates.
Before we dive into the properties of the set $\Delta$, we must show that this set is uniquely defined given its assumptions.

\begin{restatable}{lemma}{lemmatoleranceuniqueness}\label{lemma:tolerance_uniqueness}
Given an LTS $T$, controller $f$, and property $P$, there is a unique $\Delta(T,f,P)$ that satisfies the conditions in Def.~\ref{def:tolerance}. 
\end{restatable}

\begin{example}\label{example:tolerance}
Consider the same setup as in Example~\ref{example:perturbation-relation}.
The four perturbations in Example~\ref{example:perturbation-relation} are tolerable.
Therefore, they must be represented in $\Delta$ as stated in condition~(2) in the definition of $\Delta$.
At this moment, we simply provide $\Delta$ for this example and in Section~\ref{sect:tolerance-invariance} we provide the formal results on efficiently obtaining this $\Delta$.
The set $\Delta$ in this example is given by $\Delta = \{Q\times Act \times Q\setminus \{(1,b,3),(2,b,3),(4,b,3)\}\}$.
Intuitively, $\Delta$ is defined by a single perturbation set that contains all possible transitions except the ones from states $1,2,4$ to state $3$ with action $b$.
Adding any of these missing transitions make the perturbation set in $\Delta$ to be intolerable.
The perturbed system defined by this perturbation is depicted in Fig.~\ref{fig:T-cx-inv-ctr3} where we highlight the new transitions in blue \footnote{For simplicity, we do not show the transitions starting in state $3$}.
Any other tolerable perturbation is represented in $\Delta$.
For example, perturbations $d_1,d_2 \subseteq d = Q\times Act \times Q\setminus \{(1,b,3),(2,b,3),(4,b,3)\}$ which implies that $d_1,d_2\preceq_f d$.  
And although $d_3$ and $d_4$ are not subsets of $d$, it also follows that $d_3,d_4\preceq d$ since $Runs(T_{d_3}|f)=Runs(T_{d_4}|f)\subset Runs(T_d|f)$.
\end{example}
\begin{figure}[!h]
\centering  
\subfloat[LTS $T_{d}$ for $d\in \Delta$]  
{  
\centering
\includegraphics[width=0.19\columnwidth]{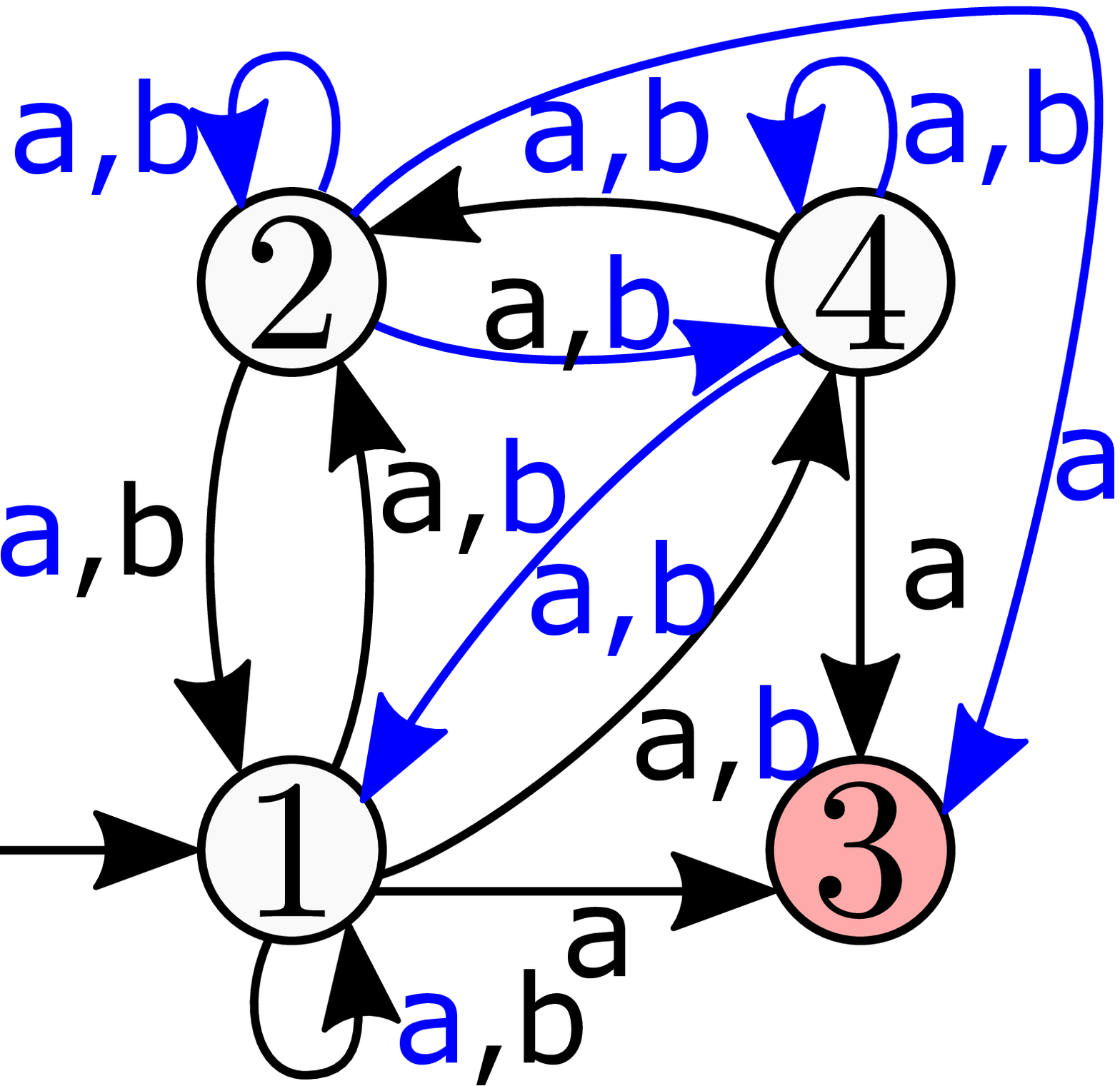}
\label{fig:T-cx-inv-ctr3}
}
\qquad
\subfloat[LTS $T_{d_1\cup d_2}$]  
{  
\centering
\includegraphics[width=0.17\columnwidth]{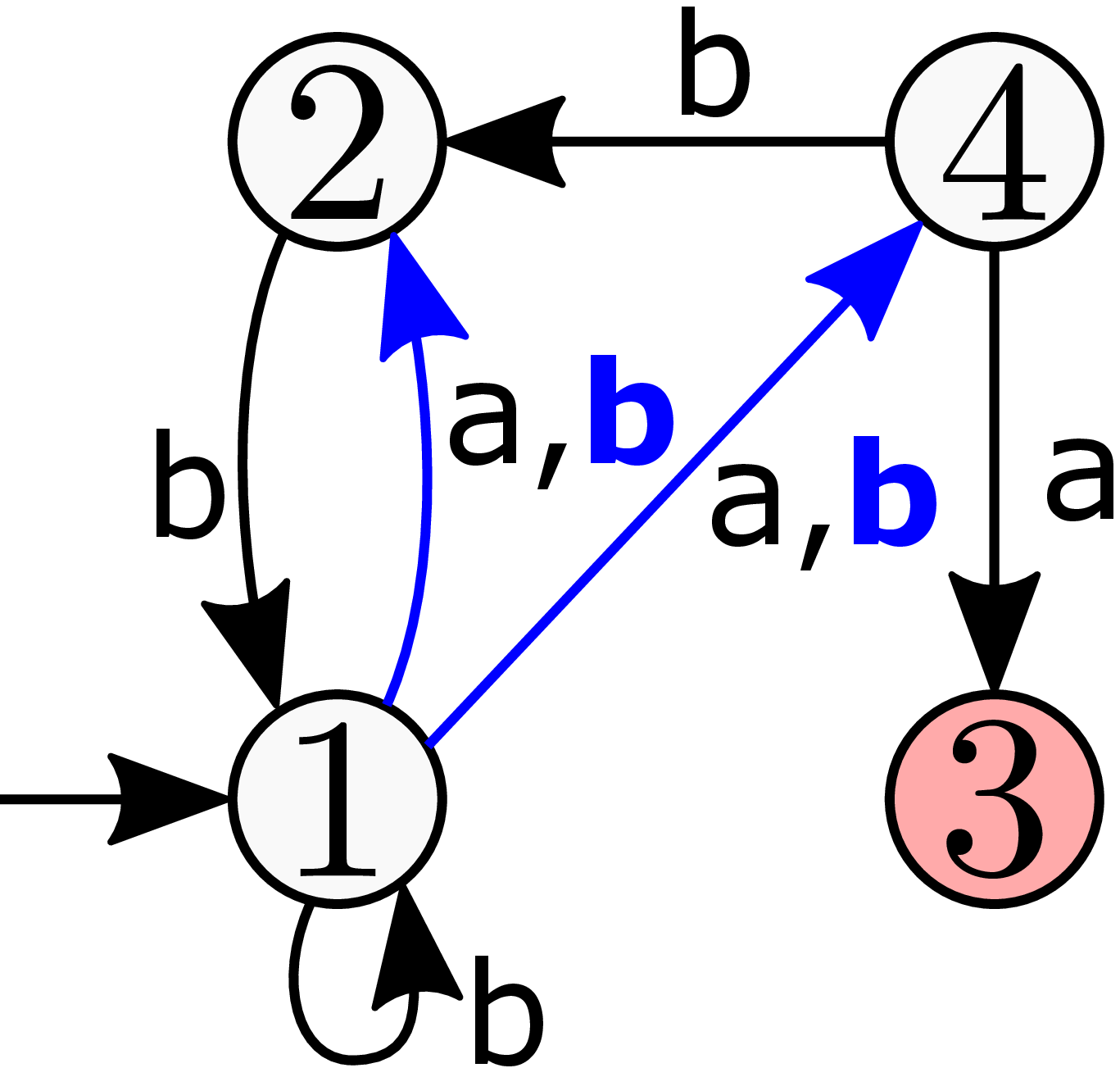}
\label{fig:T-cx-inv-ctr2}
}
\caption{LTS in Examples~\ref{example:tolerance} and~\ref{example:cx-inv-ctr}}
\end{figure}
\vspace{-1cm}
\subsection{Computing tolerance for general properties}

 The tolerance of controller $f$ is defined by the set of maximal tolerable perturbations with respect to property $P$.
 The first problem we investigate is to compute the set $\Delta$ given $T$, $f$, and $P$.

\begin{problem}\label{prob:comp-tol}
Given LTS $T$, property $P$, and controller $f$, compute $\Delta(T,f,P)$.
\end{problem}

Naively, solving Problem~\ref{prob:comp-tol} can be broken into (i) finding the set of tolerable perturbations and (ii) identifying the maximal ones within this set.
Step (i) can be reduced to verifying if system $T_{d}|f$ satisfies property $P$ for every possible perturbation $d$.
Step (ii) orders the tolerable perturbations with respect to relation $\preceq_f$.
Under mild assumptions that $P$ is a regular language and that $f$ has bounded memory, both steps (i) and (ii) are decidable.

Although this naive algorithm computes the tolerance of $f$, it will not scale for large LTS.
For this reason, we wish to investigate efficient ways to compute the set $\Delta(T,f,P)$.
Our goal is to investigate tolerance with respect to special classes of properties, e.g., invariance, safety, liveness, etc.
In the next section, we show our results for invariance properties.

\section{Tolerance with respect to invariance properties}\label{sect:tolerance-invariance}

An \emph{invariance property} $P$ for an LTS $T$ can be represented by a subset of invariant states $Q_{inv}\subseteq Q$ \cite{Baier:2008}.
Formally, a property $P$ is an invariance property if there exists an invariant set of states $Q_{inv}$ such that $P = Q_{inv}^{*}\cup Q_{inv}^{\omega}$.
For instance, $Q_{inv} = \{1,2,4\}$ in Example~\ref{example:perturbation-relation}.
An LTS satisfies an invariance property if and only if the LTS only reaches states in $Q_{inv}$ \cite{Baier:2008}.
For convenience, we assume that the invariant set of states always contains the set of initial states.

\subsection{Supremum tolerable perturbation} 
Usually when dealing with invariance properties, one can show the existence of a single supremum element that satisfies the desired investigated property.
In our scenario, we want to show that the tolerance of $f$ with respect to an invariance property is represented by a unique tolerable perturbation, i.e., $|\Delta(T,f,P)| = 1$.
Although $\Delta$ in Example~\ref{example:tolerance} has a single element, the following counterexample illustrates that in general $|\Delta(T,f,P)| \geq 1$. 

\begin{example}\label{example:cx-inv-ctr}
Consider the setup of Example~\ref{example:perturbation-relation} with LTS defined by Fig.~\ref{fig:T-cx-inv-ctr} and $Q_{inv} = \{1,2,4\}$, but under control of the following controller: $f(1214)=\{a\}$ and $f(x_{0\dots n})= \{b\}$ for any $x_{0\dots n} \in Q^+$ other than $1214$.
Perturbations $d_1 = \{(1,b,2)\}$ and $d_2 = \{(1,b,4)\}$ remain tolerable with respect to this new controller.
And although these perturbations are tolerable, their union is not tolerable since path $1214$ becomes feasible in $T_{d_1\cup d_2}$ as seen in Fig.~\ref{fig:T-cx-inv-ctr2}.
The size of $\Delta(T,f,P)$ must be at least two since we cannot combine $d_1$ and $d_2$ as a single tolerable perturbation that generates the behavior of $T_{d_1}|f$ and $T_{d_2}|f$.
\end{example}

\subsubsection{Invariant controllers}
The counterexample in Example~\ref{example:cx-inv-ctr} sheds light on the problem of the controller $f$ selecting ``bad'' control decisions for paths outside of $Paths_{fin}(T|f)$.
This problem can be easily fixed for invariance properties by introducing the notion of \emph{invariant control actions} and \emph{invariant controllers}.
\begin{definition}\label{def:inv-actions}
Let an LTS $T$ and an invariance property $P$ with invariant set of states $Q_{inv}$ be given. 
The set of \emph{invariant control actions} is defined as $A_{inv}(q) := \{a\in Act\mid Post_T(q,a)\subseteq Q_{inv}\}$ if $q\in Q_{inv}$ and $A_{inv}(q) := \emptyset$ if $q\notin Q_{inv}$.
Moreover, we say that $f$ is an \emph{invariant controller} with respect to $T$ and $P$ if $f(x_{0\dots n}) \subseteq A_{inv}(x_n)$ for any sequence $x_{0\dots n}\in Q^+$.
\end{definition}

Informally, invariant control actions characterize the ``good'' actions with respect to LTS $T$ and invariance property $P$.
Therefore, all invariant controllers satisfy invariance property $P$ as stated in Lemma~\ref{lemma:inv-ctr}.

\begin{restatable}{lemma}{leminvctr}\label{lemma:inv-ctr}
Any invariant controller with respect to $T$ and $P$ satisfies $T|f\models P$.
\end{restatable}

Although it seems that invariant controllers are restrictive, their actions only assume LTS $T$ and invariance property $P$.
Thus, one expects that all actions the controller takes satisfy the invariant control actions constraint.

\subsubsection*{Tolerance of invariant controllers}

Under the assumption of invariant controllers, the tolerance of a given controller $f$ is completely defined by a \emph{unique tolerable perturbation}, i.e., $|\Delta(T,f,T)| = 1 $ for any invariant controller $f$.
We formalize this statement in the following theorem.

\begin{restatable}{theorem}{theoinvunique}\label{theo:inv-unique}
Let LTS $T$, invariance property $P$ with invariant set of states $Q_{inv}$, and invariant controller $f$ be given.
It follows that $|\Delta(T,f,P)| = 1$ and its unique element is defined as 
$$\lceil f \rceil:= (Q\times Act \times Q)\setminus \{(q,a,q')\in Q_{inv}\times Act\times Q\setminus Q_{inv}\mid a\in F(q)\}$$
where $F(q) :=\{a\in Act\mid \exists x_{0\dots n}\in Paths_{fin}(T_{\Omega}|f).\ a \in f(x_{0\dots n})\wedge q=x_n\}$ and where $\Omega := Q_{inv}\times Act \times Q_{inv}$.
\end{restatable}

The first part of Theorem~\ref{theo:inv-unique} states that the tolerance of $f$ has a \emph{single perturbation}, i.e., there exists a supremal element within the set of tolerable perturbations with respect to $\preceq_f$.
The second part of this theorem characterizes this unique perturbation, i.e., $\lceil f \rceil$.
This element is defined by removing transitions that are not tolerable from the set of all possible transitions.
For this reason, the removed transitions are from states in $Q_{inv}$ to states outside of $Q_{inv}$.

Discussing the set difference in more detail, the function $F(q)$ restricts attention to paths in $T_{\Omega}|f$.
Recall that relation $\preceq_f$ prioritizes the behavior generated by a perturbed controlled system, i.e., $T_d|f$.
The tolerable perturbation $\Omega$ is selected since it can make every states in the the invariant set reachable, i.e., more behavior can be generated. 
Next, we investigate which actions the controller uses in the invariant states reached in $T_{\Omega}|f$.
Intuitively, if the controller uses action $a$ in a reachable invariant state $q$, then transitions $\{q\}\times \{a\}\times Q\setminus Q_{inv}$ are not tolerable and removed from $\tol{f}$.

\begin{example}\label{example:tol-inv}
We return to Example~\ref{example:tolerance} to discuss Theorem~\ref{theo:inv-unique}.
The LTS $T$ is depicted in Fig.~\ref{fig:T-cx-inv-ctr}, the invariance property $P$ is defined by the set $Q_{inv} = \{1,2,4\}$, and invariant controller $f$ is defined as $f(q) = \{b\}$ if $q\in Q_{inv}$ and $f(3) = \emptyset$.
It follows that $F(q)$ is equal to $f(q)$ for any $q\in Q$.
Intuitively, the function $F$ defines which actions the controller uses in each invariant state, e.g., action $b$ is used in state $1$.
Since the controller uses action $b$ in state $1$, the system is not tolerant if it is perturbed by transition $(1,b,3)$.
Similarly, action $b$ is also used in states $2$ and $4$ which results in $\tol{f} = Q\times Act\times Q\setminus \{(1,b,3),(2,b,3),(4,b,3)\}$.
Figure~\ref{fig:T-cx-inv-ctr3} depicts the perturbed system $T_{\tol{f}}$.
\end{example}

\subsection{Computing tolerance for invariance properties}
Problem~\ref{prob:comp-tol} investigates the computation of the set $\Delta$ for a general property $P$.
We strengthen Problem~\ref{prob:comp-tol} to invariance properties as to use the results of Theorem~\ref{theo:inv-unique}.
\begin{problem}\label{prob:comp-tol-inv}
Given LTS $T$, invariance property $P$, and invariant controller $f$, compute $\Delta(T,f,P)$.
\end{problem}

As in the solution of Problem~\ref{prob:comp-tol}, Problem~\ref{prob:comp-tol-inv} is decidable when $f$ has bounded memory as the controlled system $T|f$ is representable by an LTS.
By an abuse of notation, the LTS representation of $T|f$ is also denoted by $T|f$.

For invariance property $P$ and invariant controller $f$, $\Delta(T,f,P)$ is uniquely defined by $\tol{f}$. 
To compute the set $\tol{f}$, we need to characterize the function $F(q)$, which involves a reachability analysis of the perturbed system $T_{\Omega}|f$, i.e., $\exists x_{0\dots n}\in Paths(T_{\Omega}|f)$.
Therefore, we can use standard reachability algorithms to compute $\tol{f}$, e.g., see algorithms 48 and 49 in \cite{Baier:2008}.
These algorithms are linear in the number of states and transitions of the LTS in analysis.
Considering that the number of transitions in $T_{\Omega}|f$ is much larger than the number of states, the computation time of the $\Delta$ set is quadratic in the number of states and memory-size of controller $f$.

\begin{restatable}{proposition}{propalgoanalysis}
In the worst case, the effort involved in solving Problem~\ref{prob:comp-tol-inv} is $\mathcal{O}(|Q|^2M^2|Act|)$ where $M$ is the finite memory-size used by controller $f$.
\end{restatable}

\subsection{The least and most tolerant invariant controllers}

There is an inherent trade-off between tolerance and the restriction controller $f$ imposes on LTS $T$. 
Controllers that are more {\em permissive}~\cite{Bernet:2002,Lafortune:2008}, i.e., that allow more behaviors on $T$, are necessarily less tolerant and vice-versa.
The two extremes of this trade-off are the least and the most tolerant invariant controllers.
Formally, we search for controllers $f_1$ and $f_2$ that satisfy $\tol{f_1}\subseteq \tol{f} \subseteq \tol{f_2}$ for any other invariant controller $f$.
\begin{definition}\label{def:ctr-inv-empty}
We define controllers $f^{inv}$ and $f^{\emptyset}$ with respect to LTS $T$ and invariance property $P$ as: $f^{inv}(q) := A_{inv}(q)$ and $f^{\emptyset}(q) := \emptyset$ for any $q\in Q$.
\end{definition}

The controller $\finv$ selects the invariant control actions of each state as its decision whereas $\fem$ disables every action.
It follows that $\finv$ is the least tolerant controller whereas $\fem$ is the most tolerant among all invariant controllers.

\begin{restatable}{theorem}{theolargsmallcontroller}\label{theo:largest-smallest-controllers}
Let LTS $T$ and invariance property $P$ be given.
For any invariant controller $f$ with respect to $T$ and $P$, it follows that $\tol{\finv}\subseteq \tol{f} \subseteq \tol{\fem}$.
\end{restatable}

Intuitively, controller $\fem$ blocks the system from executing any action regardless of the perturbation.
For this reason, $\fem$ provides the largest tolerance set at the trade-off of blocking any run to be generated.
On the other hand, controller $\finv$ allows the maximum possible set of runs of $T$ that do not violate property $P$.
Consequently, $\finv$ is more susceptible to perturbations and provides the smallest tolerance set at the trade-off of allowing more behavior to be generated.
We leave for future work to better investigate this trade-off for controllers other than $\fem$ and $\finv$.

\subsection{Computing tolerable controllers}

Theorem~\ref{theo:largest-smallest-controllers} shows the existence of the most and the least tolerant invariant controllers.
In this section, we study a synthesis problem that exploit the spectrum of controllers with tolerance levels in between these controllers.
This synthesis problem explores the idea of obtaining a controller with a minimum desired level of tolerance.
However, specifying only a minimum level of tolerance is  insufficient as there might exist multiple controllers that satisfy this requirement.
Additional to a minimum level of tolerance, we search for a controller with tolerance as close as possible to the desired level as we do not want the controller to be unnecessarily restrictive.

\begin{problem}\label{prob:tolerant-ctr-close}
Given LTS $T$, invariance property $P$, and perturbation set $d$, synthesize controller $f^*$ such that (i) $d\subseteq \tol{f^*}$; and (ii) $\forall f$ that satisfies (i), $\tol{f^*}\subseteq \tol{f}$.
\end{problem}
Condition (i) requires that the tolerance of $f^*$ is at least $d$. 
Condition (ii) states that the tolerance of $\tol{f^*}$ is as close as possible to $d$, where closeness is defined by set inclusion.
Intuitively, any controller $f$ that tolerates $d$ is more tolerable than $f^*$, $\tol{f^*}\subseteq \tol{f}$.

The solution to Problem~\ref{prob:tolerant-ctr-close} comes from Theorem~\ref{theo:largest-smallest-controllers}.
The controller $\finv$ is the least tolerant controller with respect to $T$.
The solution to Problem~\ref{prob:tolerant-ctr-close} is defined by the least tolerant controller with respect to $T_d$.
\begin{restatable}{proposition}{propcontrollersynt} \label{prop:solution-tol-ctr-close}
Consider the setup in Problem~\ref{prob:tolerant-ctr-close}.
Controller $f^d$ defined as $f^d(q) := \{a\in Act\mid Post_{T_d}(q,a)\subseteq Q_{inv}\}$ is a solution to Problem~\ref{prob:tolerant-ctr-close}.
\end{restatable}
Note that the definition of $f^d$ is almost the same as the definition of $\finv$.
Their only difference is that $f^d$ is defined over $T_d$ whereas $\finv$ is defined over $T$.

\section{Case study} \label{sect:experiments}

In this section, we apply and discuss our notion of tolerance to the surveillance example described in Section~\ref{sect:motivating}.
We show how our definition provides useful information about the tolerance of controllers.
First, we show that our definition captures types of environmental perturbations that occur in practice.
Next, we compare the tolerance of two different controllers and the physical meaning of their tolerance.
Our analysis shows that our formal tolerance notion complies with the informal intuition about the tolerance of these controllers.
We also show that these controllers can be computed via Problem~\ref{prob:tolerant-ctr-close}.

We implement a tool to compute the tolerance of invariant properties as well as controllers $\fem,\finv$, and $f^d$ on top of the MDESops tool\footnote{\url{https://gitlab.eecs.umich.edu/M-DES-tools/desops}}.
Our evaluation was done on a Ubuntu 20.04 LTS OS machine with $3.2$GHz CPU and $32$GB memory.
Our implementation is available on GitHub \footnote{\url{https://github.com/romulo-goes/tolerancetool}} and the artifact in \cite{Meira-Goes:2021-sw-artifact} reproduces our results in this paper. 
We analyze the performance of our tool by scaling the motivating example as well as comparing to the naive algorithm to obtain tolerance.
\begin{figure}[!h]
\centering  
\subfloat[LTS representation of $T|f_1$]  
{  
\centering
\includegraphics[width=0.16\columnwidth]{LTS-f1}
\label{fig:LTS-f1}
}
\qquad
\subfloat[LTS representation of $T|f_2$]  
{  
\centering
\includegraphics[width=0.34\columnwidth]{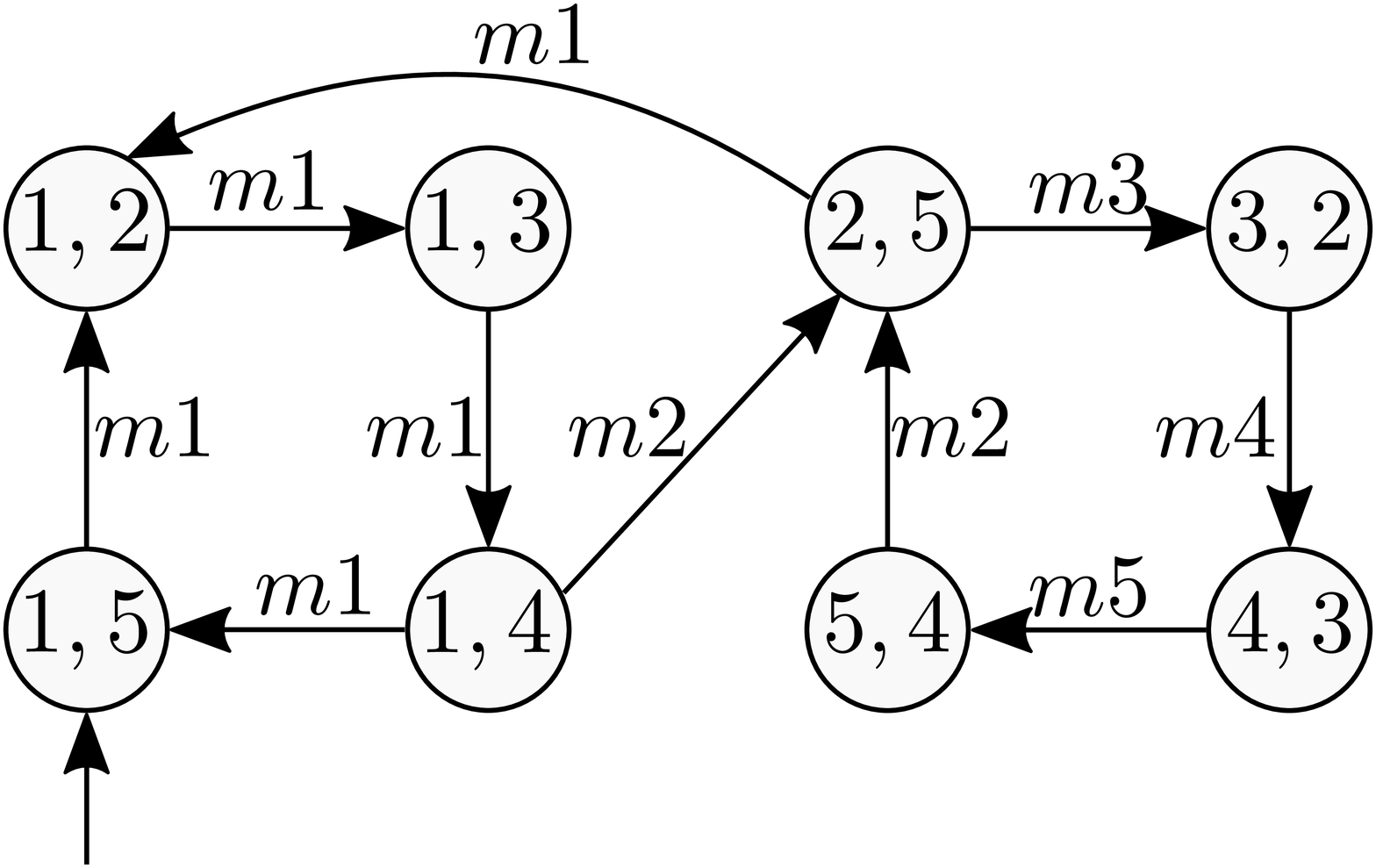}
\label{fig:LTS-f2}
}
\caption{$Ego$ under control of $f_1$ and $f_2$}
\label{fig:reach-LTS}
\end{figure}
\subsection{Models and property}

Example~\ref{example:motivating} describes how the surveillance example is modeled as an LTS.
The invariance property is defined by $Q_{inv} = Q\setminus \{(2,2),(3,3),(4,4),(5,5)\}$, i.e., $srv$ captures $ego$.
Next, we define two controllers that satisfy this invariance property.
First, we consider controller $f_1$ to be the one described in Example~\ref{example:controller} where it maintains $ego$ in location $1$.
On the other hand, controller $f_2$  ensures that $ego$ visits all locations without being captured by $srv$.
Formally, $f_2$ is defined as follows: $f_2(q) = f_1(q)$ if $q\in Q\setminus \{(1,4),(2,5)\}$, $f_2(1,4) = \{m_1, m_2\}$, and $f_2(2,5) = \{m1,m3\}$.
Figure~\ref{fig:reach-LTS} shows the LTS representations of $T|f_1$ and $T|f_2$.

\subsection{Computing the tolerance level}

We use our tool to compute the tolerance level for both controllers $f_1$ and $f_2$.
Note that LTS $T$ has $20$ states, $5$ actions, $60$ transitions, and the invariance set $Q_{inv}$ has $16$ states.
The tolerance level $\tol{f_1}$ has $1936$ transitions for which $1876$ are new transitions with respect to the transition relation $R$.
On the other hand, $\tol{f_2}$ has $1928$ transitions where $1868$ are new transitions.
In both cases, it takes about $8ms$ to compute the tolerance level.
Since every control action of $f_2$ is a subset of the corresponding one selected by $f_1$, it follows that $\tol{f_2}\subset \tol{f_1}$.
In comparison, the most tolerant controller $\fem$ characterized by $\tol{\fem} = Q\times Act\times Q$ has $2000$ transitions, i.e., the transition relation is complete.
The least tolerant controller $\finv$ has a tolerance level $\tol{\finv}$ with $1716$ transitions.

\subsection{Comparing controllers}

Controllers $f_1$ and $f_2$ select the same control decisions in all states except in states $(1,4)$ and $(2,5)$.
In these two states, controller $f_2$ allows $ego$ to venture closer to the building. 
Therefore, controller $f_1$ should be more tolerant than controller $f_2$.
This intuition is confirmed by our notion of tolerance where $\tol{f_2}\subset \tol{f_1}$, i.e., controller $f_1$ tolerates more perturbations than $f_2$.

\begin{figure}[!h]
\centering  
\subfloat[Partial $T_{\tol{f_1}}|f_1$]  
{  
\centering
\includegraphics[width=0.25\columnwidth]{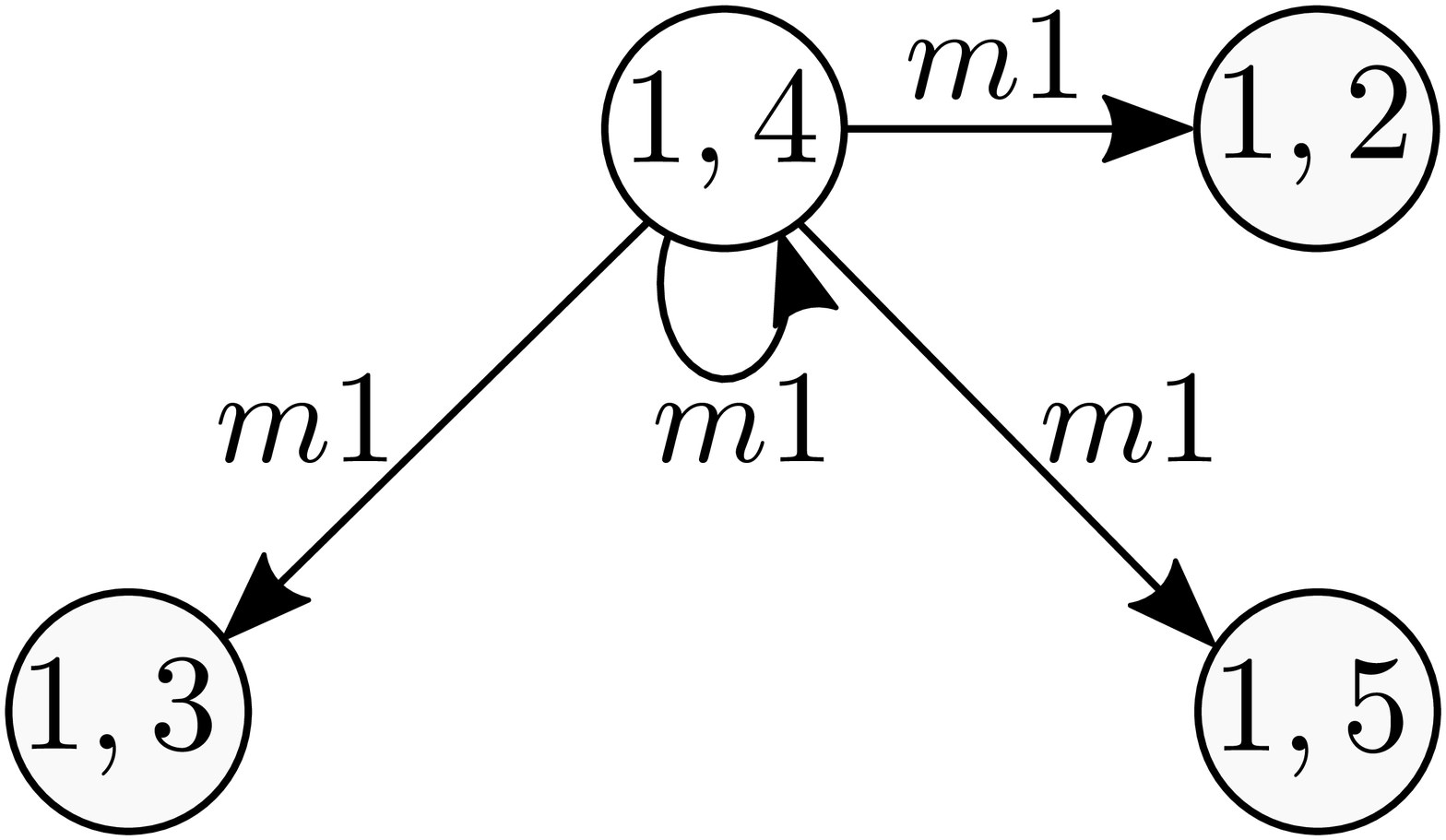}
\label{fig:LTS-tol-f1}
}
\qquad
\subfloat[Partial $T_{\tol{f_2}}|f_2$]  
{  
\centering
\includegraphics[width=0.25\columnwidth]{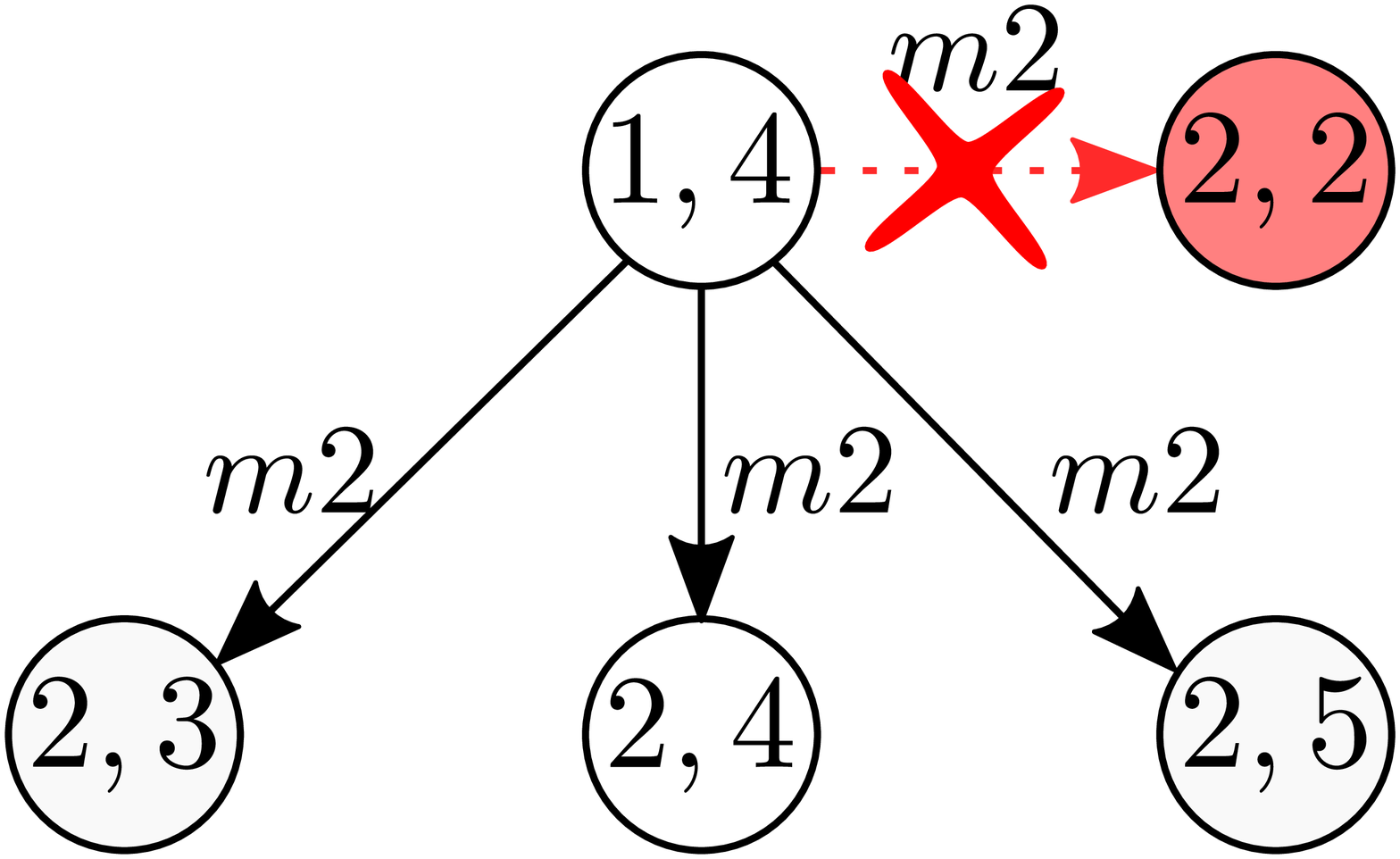}
\label{fig:LTS-tol-f2}
}
\caption{Tolerance of controllers $f_1$ and $f_2$ with respect to $srv$ perturbations}
\label{fig:LTS-tol-ctr}
\end{figure}

Figure~\ref{fig:LTS-tol-ctr} helps us explain the difference between these two controllers with respect to $srv$ strategy perturbation.
In the case of $f_1$, this controller selects action $m1$ in state $(1,4)$.
As shown in Fig.~\ref{fig:LTS-tol-f1}, $f_1$ tolerates any perturbation from $srv$. 
For example, it tolerates $srv$ being faster than expected, e.g., transition $(1,4), m1, (1,2)$ is tolerated.
On the other hand, $f_2$ cannot tolerate this type of perturbations as shown in Fig.~\ref{fig:LTS-tol-f2}.
Controller $f_2$ can select action $m2$ in state $(1,4)$, but it does not tolerate transition $(1,4), m2, (2,2)$.

Controllers $f_1$ and $f_2$ can be synthesized by solving Problem~\ref{prob:tolerant-ctr-close}, i.e., by defining a minimum level of tolerance these controllers need to achieve.
In the case of controller $f_1$, this minimum level of tolerance is informally defined by not constraining $srv$'s movements.
In other words, it is required that $srv$ should match any decision $ego$ takes when it goes to states $2,\dots, 5$, e.g., $ego$ must tolerate transitions $\big((1,4),m2,(2,2)\big)$, $\big((1,2),m2,(2,2)\big)$, $\big((1,4),m4,(4,4)\big)$, etc.
Since controller $f_1$ is too restrictive, we relax its minimum level of tolerance to obtain controller $f_2$.
Details about these sets can be found in our implementation.

As we mention in Section~\ref{sect:tolerance-invariance}, there is a trade-off between tolerance and the behavior allowed by the controller.
In this example, controller $f_2$ ventures to locations $2,\dots, 5$ whereas controller $f_1$ maintains $ego$ in location $1$.
The benefits of $f_1$ being more tolerant than $f_2$ comes at the cost of being less permissive.

\subsection{Performance analysis}

To test the performance of our tool, we scale the surveillance example by adding more locations as well as more surveillance drones.
More details about the modified surveillance models is provided in our GitHub repository.
Table~\ref{table:summary-performance} summarizes the evaluation of our tool.
The tolerances $\tol{\finv}$ in these examples are almost a complete transition relation, i.e., $\tol{\finv}\approx Q\times Act\times Q$.
Since our tool is built as proof-of-concept, it ran out of memory and it could not compute the tolerance for the system with $10$ locations and $3\,\ srv$ drones.
The complete transition relation for this system has $262\ 144\ 000$ transitions.
As part of future work, we plan to improve our tool by symbolic encoding of the LTS, e.g., using OBDD \cite{Bryant:1992}.
\renewcommand{\arraystretch}{1.2}
\setlength{\tabcolsep}{5pt}
\begin{table}
\begin{center}
\begin{tabular}{@{}|c|c|c|c|c|c|@{}}
\hline
System & $|Q|$ & $|Act|$ & $|R|$ & $|\tol{\finv}|$ & time \\\hline 
$1\,ego$, $1\,srv$, 5 locations & $20$ & $5$ & $60$ & $1\,716$ & $0.01$ sec \\\hline
$1\,ego$, $1\,srv$, 10 locations & $80$ & $10$ & $272$ & $59\,030$ & $0.46$ sec \\\hline  
$1\,ego$, $2\, srv$, 10 locations & $640$ & $10$ & $2\,176$ & $3\,618\,978$ & $30.88$ sec \\\hline
$1\,ego$, $3\,srv$, 10 locations & $5\,120$ & $10$ & $17\,408$ & out of memory & $-$ sec \\\hline
 \end{tabular} 
\caption{Tolerance of $\finv$: scalability of the case study}
\label{table:summary-performance}
\end{center}
\end{table}

We also compare the algorithm to solve Problem~\ref{prob:comp-tol-inv} with the time to verify every possible pertubed system $T_d|f$ as described in the naive algorithm to solve Problem~\ref{prob:comp-tol}.
Since the naive algorithm verifies $T_d|f\models P$ for every perturbation $d\subseteq Q\times Act\times Q$, it is infeasible to use the surveillance example since there are $2^{1940}$ systems to verify.
For this reason, we make this comparison using a modified version of the LTS shown in Fig.~\ref{fig:T-cx-inv-ctr}. 
We use FuseIC3, an off-the-shelf tool that efficiently verifies a family of LTS by reusing information from earlier verification runs \cite{Dureja:2017}.
Table~\ref{table:summary-comparison} summarizes the results of our comparison.
Although FuseIC3 efficiently verifies a large family of LTS, it was not developed to solve Problem~\ref{prob:comp-tol-inv}.
On the other hand, our algorithm directly computes the tolerance of the LTS using the results of Theorem~\ref{theo:inv-unique}.
\begin{table}
\begin{center}
\begin{tabular}{@{}|c|c|c|c|c|c|@{}}
\hline
$|Q|$ & $|Act|$ & $|R|$ & \# perturbations & Our method &  FuseIC3  \\\hline 
$4$ & $2$ & $22$ & $2^{10}$ &\textbf{0.001 sec} & 1.5 sec \\\hline
$4$ & $2$ & $17$ & $2^{15}$ &\textbf{0.001 sec} & 48.1 sec  \\\hline
\end{tabular} 
\caption{Comparison with FuseIC3}
\label{table:summary-comparison}
\end{center}
\end{table}

\vspace*{-1.8cm}
\section{Related work} \label{sect:related}
Several works investigated notions of robustness, tolerance, and resilience for discrete transition systems by quantifying perturbation via cost functions, metrics, etc. \cite{Bloem:2014,Bloem:2009,Chaudhuri:2011,Henzinger:2014,Majumdar:2011,Neider:2020,Roopsha:2013,Tabuada:2012}.
Our notion of tolerance is qualitative as it captures the set of perturbations for which the controller guarantees the property and avoids the need of external cost functions over the discrete transition system.
With respect to qualitative robustness notions, the work in \cite{Topcu:2012} investigated synthesizing controllers robust against perturbation sets specified by the designer.
Our notion of tolerance defines these perturbation sets for each controller.
In \cite{Tabuada:2016}, authors presented the notion of robust linear temporal logic (rLTL) which extends the binary view of LTL to a 5-valued semantics to capture different levels of property satisfaction.
This work is tangent to ours as it focuses on specifying robustness.

Of particular relevance to this paper are the works in \cite{Kang:2020,Zhang:2020}, which inspired our notion of tolerance. 
The notion of robustness presented in \cite{Kang:2020,Zhang:2020} is only semantically defined.
In \cite{Zhang:2020}, the environmental perturbation is captured by a set of input traces the software system accepts.
Perturbations in \cite{Kang:2020} are connected to different attack threats models for software systems.
In our work, we define the syntax of perturbations as additional transitions in the environment model.
Moreover, the semantics in our perturbation definition differs from those in \cite{Kang:2020,Zhang:2020}.

There also exist a vast literature on robust control in discrete event systems \cite{Alves:2019,Cury:1999,Lin:1993,Lin:2014,Lin:2019cdc,Meira-Goes:2019cdc,Meira-Goes:2021tac-robust,Rohloff:2012,Takai:2004,Wang:2016,Young:1995}. 
Robustness in \cite{Alves:2019,Lin:2014,Lin:2019cdc,Meira-Goes:2019cdc,Meira-Goes:2021tac-robust,Rohloff:2012,Wang:2016} are specific to communication delays, loss of information, or deception attacks.
Our notion of tolerance represents model uncertainty, which can attributed to unreliable communication channels in the controlled system.
Robustness against model uncertainty is tackled in the works of \cite{Cury:1999,Lin:1993,Takai:2004,Young:1995}.
Although our notion of tolerance resembles the ones in \cite{Cury:1999,Lin:1993,Takai:2004,Young:1995}, the semantics of our work differs from theirs as we use a different modeling formalism. 

The description of the general algorithm to compute $\Delta(T,f,P)$ for any property $P$ connects our work to the work on verifying software product lines (SPL) described as \emph{feature transitions systems} (FTS) \cite{Classen:2010,Classen:2013}.
However, verifying FTS has exponential worst-case time complexity even for invariance properties whereas our method has quadratic worst-case time complexity.
Modal transition systems (MTS) \cite{Larsen:1988,Huth:2001} can also be used to describe a family of LTS, where transitions can be mandatory (must transitions) or optional (may transitions).
In \cite{DIppolito:2012}, a controller realizability problem is studied for an environment modeled by MTS, where a controller satisfies a property in all, some, or none of the LTS family.
Our notion of controller explicitly computes which systems in the LTS family satisfy the property.

The last body of work related to this paper is the work on fault-tolerance.
Fault-tolerance has been studied in the context of distributed systems \cite{Gartner:1999,Lynch:1996,Pease:1980}.
The work in \cite{Bonakdarpour:2008,Cheng:2011,Ebnenasir:2005,Girault:2009} focuses on synthesis of fault-tolerant programs by retrofitting initial fault-intolerant programs.
These works focus on specific types of fault models, whereas our tolerance notion upper-bounds the perturbations (faults) the controller tolerates.
In the context of control of discrete transition systems, \cite{Paoli:2005} proposes a fault-tolerance framework for a control system.
However, this work requires the fault model to be explicitly specified.

\section{Conclusion}\label{sect:conclusion}

In this paper, we introduced a new notion of tolerance against environmental perturbations.
This notion defines an upper bound on the possible environmental perturbations that a controller tolerates with respect to a desired property.
We provided a general technique to compute this tolerance level for general properties modeled as regular languages over finite strings as well as a more efficient technique specifically for invariance properties.
We also investigate the problem of synthesizing an invariant controller that achieves a given minimum threshold of tolerance.

\vspace*{-.5cm}
\subsubsection{Limitations and future work:}
Our notion of tolerance is syntactically defined by additional transitions and semantically defined by the controlled behavior generated by these additional transitions.
However, the additional transitions and new controlled behavior need to be analyzed by a designer as to explain them within the context of the model.
We leave to future work to bridge this gap between the syntax of our notion of tolerance with the context of the model to provide tolerance explanations to the designer.
Another limitation is that our notion of tolerance can only be efficiently computed for invariance properties.
As part of future work, we will devise more efficient techniques for properties different than invariance.
At the end of Section~\ref{sect:experiments}, we discussed the trade-off between tolerance and permissiveness of two different controllers.
Due to space limitations, we did not provide an in-depth discussion, and this is left as part of future work. 

\bibliographystyle{splncs04}
\bibliography{bib_romulo.bib}

\longfalse
\longtrue
\iflong 
\appendix
\section*{Appendix}

\propmaxdev*
\begin{proof}
Direct proof. Without loss of generality, we assume that $d_1\cap R = d_2\cap R = \emptyset$.
Based on the transitions being used/active in the controlled system, we partition the perturbation in two sets: active transitions, and not active transitions.
We show that set of active transitions in $d_1$ and $d_2$ are equal since $Runs(T_{d_1}|f) = Runs(T_{d_2}|f)$.
Lastly, we present that the union of $d_1$ and $d_2$ satisfies $d_1,d_2\preceq_f d_1\cup d_2$ and $Runs(T_{d_1\cup d_2}|f) = Runs(T_{d_1}|f) = Runs(T_{d_2}|f)$. 

Let $d_i^{act} := \{(q,a,q')\in d_i\mid \exists x_{0}a_0\dots x_n\in Runs(T_{d_i}|f).\ (q,a,q') = (x_{n-1},a_{n-1},x_{n})\}$ for $i\in \{1,2\}$ denote the set of active transitions.
By construction, it follows that $Runs(T_{d_i^{act}}|f) = Runs(T_{d_i}|f)$ for $i\in\{1,2\}$. 

We show by subset inclusion that $d_{1}^{act}=d_{2}^{act}$.
First, we demonstrate that $d_1^{act} \subseteq d_2^{act}$.
By the definition of $d_1^{act}$ and the condition $Runs(T_{d_1}|f) = Runs(T_{d_2}|f)$, we have
\begin{align}
(q,a,q')\in d_1^{act} &\Rightarrow \exists x_{0}a_0\dots x_n\in Runs(T_{d_1}|f).\ (q,a,q') = (x_{n-1},a_{n-1},x_{n})\\
&\Rightarrow \exists x_{0}a_0\dots x_n\in Runs(T_{d_2}|f).\ (q,a,q') = (x_{n-1},a_{n-1},x_{n}) \label{eq:implication_1}\\
& \Rightarrow (q,a,q') \in d_2\cup R\\
& \Rightarrow (q,a,q') \in d_2 \text{ as } d_1\cap R = \emptyset \label{eq:implication_2}
\end{align}
The implications in lines \ref{eq:implication_1} and \ref{eq:implication_2} provide that $(q,a,q')\in d_2^{act}$.
Using similar arguments, it follows that $d_2^{act}\subseteq d_1^{act}$.
We conclude that  $d_1^{act} = d_2^{act}$.

Finally, we have that $d_1\cup d_2 = d_1^{act}\cup (d_1\setminus d_1^{act})\cup(d_2\setminus d_1^{act})$ which implies that $Runs(T_{d_1\cup d_2}|f) = Runs(T_{d_1}|f) = Runs(T_{d_2}|f)$ and $d_1,d_2\preceq_f d_1\cup d_2$.
\end{proof}

\lemmatoleranceuniqueness*
\begin{proof}
By contradiction.
Assume that there exist $\Delta_1,\Delta_2 \subseteq 2^{Q\times Act\times Q}$ such that they satisfy conditions 1, 2, and 3 in Def.~\ref{def:tolerance} and $\Delta_1 \neq \Delta_2$.
Without loss of generality, we assume that $\exists d_1\in \Delta_1\setminus\Delta_2$.
Since $d_1\in \Delta_1$, we have that $T_{d_1}|f\models P$ as $\Delta_1$ satisfies 1.
As $\Delta_2$ satisfies 2 and $d_1\notin \Delta_2$, we have that $\exists d_2\in \Delta_2$ such that $T_{d_2}|f\models P$ and $Runs(T_{d_1}|f)\subset Runs(T_{d_2}|f)$ or $d_1\subseteq d_2$ ($d_1\preceq d_2$).
Since $d_1\in \Delta_1\setminus\Delta_2$, it follows that $Runs(T_{d_1}|f)\subset Runs(T_{d_2}|f)$ or $d_1\subset d_2$.
Back to $\Delta_1$, condition 3 implies that $d_2\notin\Delta_1$ since $d_1\in \Delta_1$ and $Runs(T_{d_1}|f)\subseteq Runs(T_{d_2}|f)$.
Furthermore, it does not exist $d\in \Delta_1$ such that $Runs(T_{d_2}|f)\subset Runs(T_{d}|f)$ or $d_2\subseteq d$, because $d_1 \in \Delta_1$ and $Runs(T_{d_1}|f)\subseteq Runs(T_{d_2}|f)\subseteq Runs(T_{d}|f)$, and $\Delta_1$ satisfies condition 3.
Consequently, the perturbation $d_2$ is a witness of the $\Delta_1$ violating condition 2, which contradicts our assumption that $\Delta_1$ satisfies conditions 1, 2, and 3.
\end{proof}

\leminvctr*
\begin{proof}
It directly follows from the definition of invariant controllers (Def.~\ref{def:inv-actions}).
\end{proof}

\theoinvunique*
\begin{proof}
We first show by contradiction that $|\Delta(T,f,P)| = 1$. 
For simplicity, we write $\Delta$ instead of $\Delta(T,f,P)$.
Since $\emptyset$ is always a tolerable perturbation, it follows that $|\Delta| \geq 1$.
Assume that $|\Delta| > 1$ and let $d_1,d_2\in \Delta$.
In the definition of $\Delta$, condition (3) states that $d_1\not\preceq d_2$ and $d_2\not\preceq d_1$.
We define $d_i^{inv}:=\{(q,a,q')\in d_i\mid q,q'\in Q_{inv} \wedge a\in A_{inv}(q)\}$ for $i \in \{1,2\}$.
By construction, the controlled system $T_{d^{inv}_i}|f$ generates the same runs as $T_{d_i}|f$ for $i \in \{1,2\}$ otherwise $d_i$ is not tolerable.
As $d_1,d_2 \in \Delta$, it must be that $d_1^{inv}$ and $d_2^{inv}$ are incomparable, otherwise $d_1\preceq d_2$ or $d_2\preceq d_1$.
Because $d_1^{inv}$ and $d_2^{inv}$ only define transitions within $Q_{inv}$ and $f$ is invariant, we have that $d_1^{inv}\cup d_2^{inv}$ is a tolerable perturbation, i.e., $T_{d_1^{inv}\cup d_2^{inv}}|f\models P$.
The perturbation $d_1^{inv}\cup d_2^{inv}$ must be represented in $\Delta$ as stated by condition (2) in Def.~\ref{def:tolerance}.
Since $d_1^{inv}$ and $d_2^{inv}$ are incomparable, the representation of $d_1^{inv}\cup d_2^{inv}$ must be different than $d_1$ and $d_2$.
Thus, there exist $d_3\in \Delta$ different than $d_1$ and $d_2$ such that $d_1^{inv}\cup d_2^{inv}\preceq d_3$.
Since the condition $Runs(T_{d_1^{inv}}|f) = Runs(T_{d_1}|f)\subset Runs(T_{d_1^{inv}\cup d_2^{inv}}|f)$, it follows that $d_1\preceq d_3$, which violates condition (3) in the definition of $\Delta$.
That is, we have two perturbation sets in $\Delta$ that are comparable via $\preceq_f$.
We reached a contradiction.

Next, we show by contradiction that $\lceil f \rceil\in \Delta$.
Assume that perturbation $d\neq \tol{f}$ satisfies $T_d|f\models P$ and $d\in \Delta$.
By construction of $\tol{f}$, the runs generated by $T_{\tol{f}}|f$ are the same as the ones generated by $T_{\Omega}|f$.
Therefore, the perturbation $\tol{f}$ is tolerable.
We have shown previously that $|\Delta| = 1$, which ensures that $\tol{f}\preceq d$ since $d\in \Delta$.
Therefore, it must be that $Runs(T_{\tol{f}}|f)\subset Runs(T_{d}|f)$ or $Runs(T_{\tol{f}}|f)= Runs(T_{d}|f)$ and $\tol{f}\subset d$.
If $Runs(T_{\tol{f}}|f)\subset Runs(T_{d}|f)$, then $d$ is not a tolerable perturbation since $Runs(T_{\tol{f}}|f) = Runs(T_{\Omega}|f)$.
If $Runs(T_{\tol{f}}|f)= Runs(T_{d}|f)$ and $\tol{f}\subset d$, then there exists a transition in $d$ that is not in $\tol{f}$ and this transition is not active in any run.
However, by the definition of $\tol{f}$, any transition in $d$ that is not in $\tol{f}$ implies that $Runs(T_{\tol{f}}|f)\subset Runs(T_{d}|f)$ and $d$ not being a tolerable perturbation.
It follows that $d$ is not a tolerable perturbation, which contradicts our assumption that $d\in \Delta$.
\end{proof}

\theolargsmallcontroller*
\begin{proof}
It follows from $F^{inv}(q) \subseteq F(q)\subseteq F^{\emptyset}$ for any invariant controller $f$ where $F^{inv}$, $F^{\emptyset}$, and $F$ are defined as in Theorem~\ref{theo:inv-unique} for controllers $\fem,\ \finv$, and $f$, respectively. 
\end{proof}

\propalgoanalysis*
\begin{proof}
It follows from the worst case effort to perform a reachability analysis over an LTS \cite{Baier:2008}.
\end{proof}

\propcontrollersynt*
\begin{proof}
We begin showing that $d\subseteq \tol{f^d}$ by a direct proof.
Recall that $\tol{f^d} = Q\times Act\times Q\setminus \{(q,a,q')\in Q_{inv}\times Act\times Q\setminus Q_{inv}\mid a\in F^d(q)\}$ where $F^d(q)= \{a\in Act\mid \exists x_{0\dots n}\in Paths_{fin}(T_{\Omega}|f^d).\ a \in f^d(x_{0\dots n})\wedge q = x_n\}$.
For any $t\in d\setminus Q_{inv}\times Act\times Q\setminus Q_{inv}$, we have that $t\in \tol{f^d}$ by definition of $\tol{f^d}$.
If $(q,a,q')\in d\cap Q_{inv}\times Act\times Q\setminus Q_{inv}$, then $a\notin f^d(q)$ by the definition of $f^d$.
It follows that $(q,a,q')\in \tol{f^d}$ since action $a$ is never used in state $q$.
Thus, we have that $d\subseteq \tol{f^d}$.

For the second part, for any invariant controller $f$ that satisfies condition (i), we must show that $\tol{f^d}\subseteq \tol{f}$.
We show it by contradiction.
Assume that there exists an invariant controller $f$ that satisfies (i) and $\tol{f}\subset \tol{f^d}$.
Without loss of generality, we assume that $\tol{f}= \tol{f^d}\setminus \{(q^*,a^*,q')\}$ for some $q^*,q'\in Q$ and $a^*\in Act$ and $f$ is memoryless.
These assumptions together with the definition of $\tol{f}$ guarantee that $q^*\in Q_{inv}$, $q'\in Q\setminus Q_{inv}$, $f(q) = f^d(q)$ for any $q\in Q\setminus\{q^*\}$, and $f^d(q^*) = f(q)\setminus \{a^*\}$.
Moreover, it follows that $Post_{T_d}(q^*,a^*)\not\subseteq Q_{inv}$ otherwise $a^*\in f^d(q^*)$.
As $f$ is invariant, the successor states of $q^*$ under action $a^*$ are a subset of $Q_{inv}$, i.e., $Post_T(q^*,a^*) \subseteq Q_{inv}$.
The last two statements ensure that $(q^*,a^*,q')\in d$.
We can conclude that $d\not\subseteq\tol{f}$ which contradicts assumption (i).
\end{proof}
\fi

\end{document}